\documentclass[11pt]{article}

\usepackage{multirow}
\usepackage[margin=1in,letterpaper]{geometry}
\usepackage{enumerate}
\usepackage{makecell}
\usepackage{amsmath,amsthm,amssymb,fancyhdr,mathrsfs,enumerate,latexsym}
\usepackage{graphicx}
\newcommand{\safeincludegraphics}[2][]{%
  \IfFileExists{#2}{\includegraphics[#1]{#2}}{%
    \fbox{\begin{minipage}[c][0.23\textheight][c]{0.95\linewidth}\centering\footnotesize Figure file not found\\\texttt{\detokenize{#2}}\end{minipage}}%
  }%
}
\usepackage{multicol}
\usepackage{longtable}
\usepackage{float}
\usepackage{caption}
\usepackage{rotating}
\usepackage{natbib}
\usepackage{url}
\usepackage[colorlinks=true,linkcolor=blue,citecolor=blue,urlcolor=blue]{hyperref}
\usepackage[section]{placeins}
\usepackage{booktabs}
\usepackage{microtype}

\newcommand{\E}{\mathbb{E}}
\newcommand{\V}{\mathbb{V}ar}
\newcommand{\hthet}{\widehat{\theta}}
\newcommand{\tthet}{\widetilde{\theta}}
\newcommand{\bthet}{\bar{\theta}}

\theoremstyle{plain}
\newtheorem{theorem}{Theorem}[section]

\theoremstyle{definition}
\newtheorem{remark}[theorem]{Remark}

\begin{document}

\begin{center}
\Large{\sc Threshold Exceedance Estimation in Spatially Correlated
Areal Data Using Maxima-Nominated Sampling}
\end{center}

\begin{center}
{Mohammad Jafari Jozani$^1$ }\\[0.5ex]
{\it $^1$ Department of Statistics, University of Manitoba, Winnipeg, MB, Canada, R3T 2N2}
\end{center}

\begin{abstract}
We study estimation of the proportion of areal units in a spatially correlated domain whose success probabilities exceed a prespecified threshold. Such problems arise in health surveillance, environmental monitoring, and social policy, where the goal is to estimate the fraction of high-risk areas. We propose a DUST-MNS design that combines maxima-nominated sampling (MNS) with the probability-proportional-to-size dependent unit sequential technique (pps-DUST), thereby promoting spatial spread while mitigating the effect of spatial autocorrelation. The design forms $n$ candidate sets of size $k$ and obtains final measurements  only from  the area judged to be at highest risk in each set, yielding $n$ measured areas from $nk$ screened candidates. Ranking may be based on expert judgment, prior surveys, or easily obtained auxiliary covariates. We derive a closed-form estimator of the exceedance probability $\theta$ based on  data from DUST-MNS design, establish its bias and variance, and show that, in the rare-to-moderate exceedance regime $\theta<\theta^\star(k)$, the proposed DUST-MNS estimator outperforms its SRS and DUST-SRS counterparts, where $\theta^\star(k)$ depends only on $k$. We also provide guidance on the choice of $k$, derive efficiency bounds under a Beta model, extend the method to imperfect ranking, and develop variance estimation and bootstrap confidence intervals. An application to county-level stroke prevalence data from CDC PLACES, using diabetes prevalence as the ranking concomitant, illustrates the proposed approach.
\end{abstract}


%======================================================================
\section{Introduction}\label{sec:intro}
%======================================================================

An important problem in spatial epidemiology, environmental surveillance, and public policy is the
estimation of an exceedance proportion and finding out  the fraction of areal units whose underlying risk exceeds a
prespecified benchmark. In many surveillance settings, this upper-tail functional is more relevant than
an overall average. Public-health agencies often need to determine how many counties, districts, or
other reporting regions lie above a high-burden threshold in order to prioritize geographically targeted
interventions, allocate prevention resources, and monitor spatial disparities. For example, in chronic
disease surveillance, one may wish to estimate the proportion of counties whose stroke prevalence
exceeds a benchmark corresponding to unusually elevated burden. \citep{GreenlundEtAl2022,ChenEtAl2026}. If $p_i$ denotes the stroke prevalence in county $i$ and $c$ is chosen as a high-burden benchmark, such as the empirical 90th percentile of the county-level distribution, then $\theta=\Pr(p_i>c)$ is the proportion of counties whose burden exceeds that benchmark.

This threshold-oriented objective differs fundamentally from mean estimation. When $c$ lies in the
upper tail, simple random sampling (SRS) typically devotes substantial effort to areas far below the
benchmark and therefore only weakly informative about $\theta$. As a result, conventional designs may
be poorly matched to the scientific question. An effective design for exceedance estimation should
concentrate measurement effort on areas likely to exceed the threshold while still preserving adequate
geographical coverage. The latter requirement is especially important in areal surveillance, since nearby
units often have similar risks and therefore contribute overlapping information.

A natural starting point for this problem is nomination sampling.  Introduced by \citet{Willemain1980}, it was developed for settings in which direct random selection is difficult but small subsets of units can be formed and ranked before full measurement is taken. Importantly, ranking need not be based on the study variable itself. In many applications, it can instead be based on expert judgment, proxy measurements, or inexpensive auxiliary covariates that are informative about the outcome of interest. This use of inexpensive concomitant information is closely related to the ranked-set sampling literature; see, for example, \citet{YuLam1997} for regression-assisted use of concomitants and \citet{NahhasWolfeChen2002} for the role of cost structure and set size in efficient ranked-set designs. In the stroke application considered later, counties can be ranked using readily available indicators of elevated stroke burden.  In the county-level stroke setting considered later, for example, diabetes prevalence is a plausible concomitant because counties with higher diabetes burden also tend to have higher stroke burden. Thus, within a small candidate set, one can often identify the county most likely to lie in the upper tail without directly observing stroke prevalence.

In its maxima-nomination form, one constructs candidate sets and fully measures only the largest unit from each set according to the auxiliary ranking. When the ranking is positively associated with the study variable, this design concentrates effort on units most likely to satisfy $p_i>c$, making it especially attractive for estimation of exceedance proportions. Existing work on nomination sampling has focused primarily on mean-, distribution-, and quantile-related functionals. For example, \citet{BoylesSamaniego1986}, \citet{TiwariWells1989}, and \citet{KvamSamaniego1993} studied estimation of distributional functionals under nomination sampling, while \citet{JafariJozaniJohnson2012} and \citet{NourmohammadEtAl2014,NourmohammadEtAl2015,NourmohammadEtAl2020} developed finite-population and inferential extensions. More recently, \citet{LoewenJafariJozani2026} used nomination sampling in robust regression. For binary or threshold-based targets, the most closely related ranked-set result is the estimation of a population proportion under ranked set sampling; see \citet{ChenStasnyWolfe2006}. Our setting differs in that the target is an exceedance proportion in a spatially correlated areal population and the nomination step is embedded within a spatially spreading design.

Despite its clear relevance, nomination sampling has not been formulated as a formal framework for estimating exceedance proportions in spatially correlated areal populations. Spatial dependence adds a further complication because neighboring areas often have positively correlated success probabilities, so naive sampling designs may overrepresent local clusters and thereby inflate variance. In addition, if nomination is used without explicit spatial control, the selected maxima may become concentrated in a restricted part of the study region, weakening spatial representativeness. The design problem therefore has two objectives: to preferentially sample areas likely to exceed the threshold while preserving spatial dispersion.

This paper develops a framework specifically for that problem. We combine maxima-nomination sampling (MNS) with a spatially spreading design by embedding it within the probability-proportional-to-size Dependent Unit Sequential Technique (pps-DUST) of \citet{Arbia1993}. Under pps-DUST, the first area is selected with probability proportional to size, and each subsequent selection probability is modified by a multiplicative penalty based on neighborhood lags from previously selected areas. Areas near those already chosen are therefore downweighted, while more distant areas remain comparatively more likely to be sampled. In this way, the design promotes spatial dispersion and reduces the tendency to oversample neighboring clusters.

The resulting DUST-MNS procedure concentrates full measurement effort on areas most likely to exceed the benchmark while preserving coverage over the study region. On the inferential side, we derive a calibrated estimator that links the exceedance probability of the nominated maxima to the population-level exceedance proportion $\theta$. We establish its bias and
variance, compare its efficiency with SRS and DUST-SRS, and characterize the regime in which the
proposed design is most advantageous. We further extend the framework to imperfect ranking and
develop variance estimation and confidence interval procedures, including bootstrap-based inference for
small or moderate sample sizes.

We illustrate the proposed method using county-level stroke prevalence data from the CDC PLACES
platform, with diabetes prevalence as the ranking concomitant. The application is motivated by the
public-health need to identify and quantify the fraction of counties with elevated stroke burden, rather
than merely summarize national average prevalence. Because PLACES provides a large geographically
indexed population with strong medical and spatial structure, it offers a natural setting in which to
study threshold-oriented areal surveillance.

 The remainder of the paper is organized as follows. Section~\ref{sec:model} describes the proposed DUST-MNS design and its implementation. Section~\ref{sec:theory} develops estimators of the exceedance probability using SRS, DUST-SRS, and DUST-MNS data and presents their basic properties. Section~\ref{sec:effcomp} establishes efficiency comparisons among the competing designs, including explicit results under a Beta model. Section~\ref{sec:imperf} extends the framework to imperfect ranking. Section~\ref{sec:application} presents the county-level stroke application and the Monte Carlo study, and Section~\ref{sec:conclusion} concludes with a brief discussion.

%======================================================================
\section{DUST-MNS design and its implementation}\label{sec:model}

Consider a study domain $\mathscr{D}$ partitioned into $N$ non-overlapping areal units. Let $N_i$ denote the total number of individuals in area $i\in\{1,\ldots,N\}$. In a selected area, suppose that outcomes are observed for $m_i$ individuals, where $1\le m_i\le N_i$, and let $Y_{ij}\in\{0,1\}$ denote the binary outcome for individual $j=1,\ldots,m_i$. Here $Y_{ij}=1$ indicates the event of interest, such as stroke, disease hospitalization, or pollution exceedance. Conditional on the area's latent success probability $p_i$, we assume
\[
Y_{ij}\mid p_i \sim \mathrm{Bernoulli}(p_i),
\]
so that the within-area observed count
\[
X_i=\sum_{j=1}^{m_i}Y_{ij}\mid p_i \sim \mathrm{Binomial}(m_i,p_i).
\]
When the selected area is fully measured, $m_i=N_i$; otherwise, $m_i$ denotes the within-area sample size.

We treat the latent success probabilities $p_1,\ldots,p_N$ as exchangeable random variables with marginal density $f(\cdot)$ on $(0,1)$, mean $p_0=\E(p_i)$, and variance $\sigma_0^2=\V(p_i)$. Spatial dependence is introduced through the correlation structure. For areas $i$ and $j$ separated by neighborhood lag $l_{ij}$, we assume
\begin{equation}\label{eq:autocorr}
\eta_{ij}=\mathrm{Corr}(p_i,p_j)=\eta_0^{\,l_{ij}}, \qquad \eta_0\in[0,1),
\end{equation}
where $\eta_0$ is the first-lag spatial autocorrelation and $l_{ij}\in\{1,2,\ldots\}$ denotes the minimum number of areal boundaries separating areas $i$ and $j$. Moran's coefficient \citep{Moran1950} may be used to estimate $\eta_0$.

Let $c\in(0,1)$ be a prespecified policy threshold. The primary inferential target is
\begin{equation}\label{eq:theta}
\theta=\Pr(p_i>c)=1-F(c),
\end{equation}
where $F(\cdot)$ is the marginal distribution function of $p_i$. In the stroke application, for example, $p_i$ is the stroke prevalence in county $i$, and if $c$ is chosen as a high-burden benchmark, such as the empirical 90th percentile of the county-level distribution, then $\theta$ is the proportion of counties whose stroke prevalence exceeds that benchmark. Because $\theta$ is a tail probability rather than a mean, its estimation calls for a design that preferentially targets the upper tail of the $p_i$ distribution while maintaining adequate spatial coverage.

To construct the proposed DUST-MNS design, we first use the pps-DUST mechanism of \citet{Arbia1993} to select a pool of $nk$ candidate areas that is both probability-proportional-to-size and spatially well dispersed. The aim is to avoid drawing many neighboring areas, since nearby units tend to have similar latent success probabilities and therefore contribute overlapping information. At the same time, the design retains the probability-proportional-to-size principle, so that areas with larger importance measure remain more likely to be sampled. The procedure is sequential. Let $M_i$ denote the size measure attached to area $i$, and let $M=\sum_{i=1}^N M_i$. If $I_j$ denotes the area selected at draw $j$, then
\[
\Pr(I_1=i)=\frac{M_i}{M}, \qquad i=1,\dots,N.
\]
For $j\ge 2$, let $s_{j-1}$ denote the set of areas already selected after the first $j-1$ draws, with complement $\bar s_{j-1}$. Each unselected area $i\in \bar s_{j-1}$ is assigned weight
\begin{equation}\label{eq:dust}
w_i^{(j)} = M_i \prod_{r\in s_{j-1}}\bigl(1-\eta_0^{\,l_{ir}}\bigr),
\end{equation}
and the next area is selected with conditional probability
\[
\Pr(I_j=i\mid s_{j-1})
=
\frac{w_i^{(j)}}{\sum_{h\in \bar s_{j-1}} w_h^{(j)}},
\qquad i\in \bar s_{j-1}.
\]
Repeating this step for $j=1,\ldots,nk$ yields the full candidate pool.

This updating rule has a simple interpretation. The factor $M_i$ preserves the probability-proportional-to-size feature, while the product term applies a multiplicative spatial penalty. If area $i$ is close to one or more previously selected areas, then some of the factors $1-\eta_0^{\,l_{ir}}$ are small, and its weight is reduced accordingly. By contrast, areas that are farther from the current sample receive less penalization and therefore remain more likely to be selected. In this sense, pps-DUST acts as a repulsive sampling device that discourages local clustering and promotes spatial coverage of the study region.

In practice, implementation requires three ingredients: a size measure $M_i$, a neighborhood structure from which the lags $l_{ir}$ are computed, and an estimate of the spatial dependence parameter $\eta_0$. In areal applications, $M_i$ may be the population size $N_i$, the number of households, or another measure of area importance. The lags $l_{ir}$ are obtained from the adjacency graph of the study region, where $l_{ir}=1$ for neighboring areas, $l_{ir}=2$ if one intermediate area separates them, and so forth. The parameter $\eta_0$ may be estimated from the data, for example by a global spatial autocorrelation measure such as Moran's coefficient. Once these ingredients are available, the pps-DUST algorithm is implemented by repeatedly updating the weights in \eqref{eq:dust}, normalizing them, and selecting the next area from the resulting distribution. Thus, before any ranking step is performed, pps-DUST produces a candidate sample of size $nk$ that is deliberately spread over space. This is especially important here, because the subsequent maxima-nomination step enriches the final  sample with areas likely to exceed the threshold, while pps-DUST prevents that enrichment from resulting in excessive spatial concentration.

The DUST-MNS design has two additional parameters: $n$, the final sample size or the number of areal units to be studied,  and $k$, the set size used for ranking. The analyst typically chooses $k$ to be small, often $k=3$, $4$, or $5$, so that within-set ranking remains reliable, and then chooses $n$ to achieve the desired precision. The total design cost is therefore based on $nk$ ranking evaluations together with $n$ full measurements. Given $n$ and $k$, the DUST-MNS design proceeds as follows.
\begin{enumerate}[Step 1.]
\item Using pps-DUST, select $nk$ candidate areas and partition them at random into $n$ sets of size $k$.
\item Within each set $i\in\{1,\ldots,n\}$, rank the $k$ areas in descending order according to an auxiliary variable that is informative about $p_i$.
\item From each set, retain only the area with the highest auxiliary value, namely the maximum nominee, for full measurement. Let $p_{[k:k],i}$ denote its true success probability.
\item In each selected area, observe the binary outcomes and compute the exceedance indicator from the resulting estimated prevalence.
\end{enumerate}

In the county-level stroke application, the auxiliary ranking variable is diabetes prevalence. Thus, within each DUST-selected candidate set, counties are ranked by diabetes prevalence, and only the county with the highest diabetes burden is chosen for full measurement of stroke prevalence. Because diabetes prevalence is strongly associated with stroke prevalence, this ranking step helps target counties most likely to exceed the high-burden threshold while the DUST component preserves spatial spread across the study region.

\begin{remark}[Why the maximum and not another order statistic]\label{rem:whymax}
Under any absolutely continuous distribution $F(\cdot)$ with $F(c)<1$, the maximum order statistic within a set of size $k$ stochastically dominates all lower order statistics; see, for example, \citet[Chapter 2]{DavidNagaraja2003}. Hence
\[
\Pr(p_{k:k}>c)\ge \Pr(p_{i:k}>c), \qquad i<k.
\]
Thus, the maximum nominee provides the largest per-set probability of observing an exceedance and therefore carries the greatest information about $\theta$ through the indicator $\mathbf{1}(p_{[k:k],i}>c)$.
\end{remark}
%======================================================================
\section{SRS- and MNS-based  exceedance probability  estimation} \label{sec:theory} 
%======================================================================

Let
$
R_n=\sum_{i=1}^n \mathbf{1}(p_{[k:k],i}>c)
$
denote the number of sets in which the nominated maximum exceeds $c$. Because the pps-DUST mechanism deliberately spreads candidate areas over the study region, the selected maxima tend to arise from geographically separated sets. We therefore assume that $p_{[k:k],1},\ldots,p_{[k:k],n}$ are approximately independent for the purpose of the following calibration. For a set of size $k$, the exceedance probability of the maximum is
\[
q_k:=\Pr(p_{[k:k]}>c)=1-(1-\theta)^k,
\]
so the natural empirical estimator of $q_k$ is
\[
\hat q_k=\frac{R_n}{n}.
\]
Although $\hat q_k$ is unbiased for $q_k$, in general $q_k\neq \theta$. This suggests estimating $\theta$ by inverting the power relationship.

%\begin{theorem}[Calibrated MNS estimator]\label{thm:unbiased}
%Define
%\begin{equation}\label{eq:thetahat}
%\hthet_{\mathrm{DUST\text{-}MNS}} = 1 - \left(1 - \frac{R_n}{n}\right)^{1/k}.
%\end{equation}
%Then $\hthet_{\mathrm{DUST\text{-}MNS}}=g_k(\hat q_k)$, where $g_k(x)=1-(1-x)^{1/k}$ and $g_k(q_k)=\theta$.
%Thus $\hthet_{\mathrm{DUST\text{-}MNS}}$ is the natural plug-in inverse of $q_k$.
%It is not exactly unbiased because $g_k$ is nonlinear. Its leading bias is
%\begin{equation}\label{eq:bias}
%\mathrm{Bias}(\hthet_{\mathrm{DUST\text{-}MNS}})
%=
%\frac{1}{2}g_k''(q_k)\frac{q_k(1-q_k)}{n}
%+ O(n^{-2})
%=
%\frac{(1/k-1)}{2kn}(1-q_k)^{1/k-2}q_k(1-q_k)+O(n^{-2}).
%\end{equation}
%\end{theorem}
%
%\begin{proof}
%Under the working independence approximation,
%$R_n\sim\mathrm{Binomial}(n,q_k)$, so $\E(\hat q_k)=q_k$ and
%$\V(\hat q_k)=q_k(1-q_k)/n$. A second-order Taylor expansion of $g_k(\hat q_k)$ around $q_k$ gives
%\[
%g_k(\hat q_k)
%=
%g_k(q_k)+g_k'(q_k)(\hat q_k-q_k)+\tfrac12 g_k''(q_k)(\hat q_k-q_k)^2+O_p(n^{-3/2}).
%\]
%Taking expectations and using $g_k(q_k)=\theta$ yields \eqref{eq:bias}.
%\end{proof}

\begin{theorem}\label{thm:unbiased}
 Let
$
R_n=\sum_{i=1}^n \mathbf{1}(p_{[k:k],i}>c)
$
denote the number of sets in which the nominated maximum exceeds $c$ and define
\begin{equation}\label{eq:thetahat}
\hthet_{\mathrm{DUST\text{-}MNS}} = 1 - \left(1 - \frac{R_n}{n}\right)^{1/k}.
\end{equation}
Then $\hthet_{\mathrm{DUST\text{-}MNS}}=g_k(\hat q_k)$, where
$
\hat q_k=\frac{R_n}{n}
$, 
$
g_k(x)=1-(1-x)^{1/k},
$
and $g_k(q_k)=\theta$. Also\begin{equation}\label{eq:bias}
\mathrm{Bias}(\hthet_{\mathrm{DUST\text{-}MNS}})
=
\frac{1}{2}g_k''(q_k)\frac{q_k(1-q_k)}{n}
+ O(n^{-2})
=
\frac{(k-1)}{2k^2}(1-q_k)^{1/k-2}\frac{q_k(1-q_k)}{n}
+O(n^{-2}).
\end{equation}
\end{theorem}

\begin{proof}
Under the working independence approximation,
$
R_n\sim\mathrm{Binomial}(n,q_k),
$
so
$
\E(\hat q_k)=q_k
$
and
$
\V(\hat q_k)=\frac{q_k(1-q_k)}{n}.
$
Also,
\[
g_k'(x)=\frac{1}{k}(1-x)^{1/k-1},
\qquad
\]
A second-order Taylor expansion of $g_k(\hat q_k)$ about $q_k$ gives
\[
g_k(\hat q_k)
=
g_k(q_k)
+g_k'(q_k)(\hat q_k-q_k)
+\frac12 g_k''(q_k)(\hat q_k-q_k)^2
+O_p(n^{-3/2}).
\]
Taking expectations,   using $\E(\hat q_k-q_k)=0$, and noting that the remainder $O_p(n^{-3/2})$ in the Taylor expansion contributes $O(n^{-2})$ after taking expectation,  yield
\[
\E\{g_k(\hat q_k)\}
=
g_k(q_k)
+\frac12 g_k''(q_k)\E\bigl[(\hat q_k-q_k)^2\bigr]
+O(n^{-2}).
\]
Therefore
\[
\E(\hthet_{\mathrm{DUST\text{-}MNS}})
=
\theta
+\frac12 g_k''(q_k)\frac{q_k(1-q_k)}{n}
+O(n^{-2}),
\]
which gives \eqref{eq:bias}.
\end{proof}

%
%\begin{proof}
%Under the working independence approximation,
%$
%R_n\sim\mathrm{Binomial}(n,q_k),
%$
%so that
%$
%\E(\hat q_k)=q_k,
%$
% and
% $
% \V(\hat q_k)=\frac{q_k(1-q_k)}{n}.
%$
%Also,
%$
%g_k'(x)=\frac{1}{k}(1-x)^{1/k-1}.
%$
%A second-order Taylor expansion of $g_k(\hat q_k)$ about $q_k$ yields
%\[
%g_k(\hat q_k)
%=
%g_k(q_k)
%+g_k'(q_k)(\hat q_k-q_k)
%+\frac12 g_k''(q_k)(\hat q_k-q_k)^2
%+\mathcal R_n,
%\]
%where
%\[
%\mathcal R_n
%=
%\frac{1}{6}g_k'''(\xi_n)(\hat q_k-q_k)^3
%\]
%for some random $\xi_n$ between $\hat q_k$ and $q_k$. Since $q_k\in(0,1)$ is fixed and $\hat q_k\to q_k$ in probability, we have $g_k'''(\xi_n)=O_p(1)$, and hence
%$
%\mathcal R_n = O_p\!\left(|\hat q_k-q_k|^3\right).
%$
%Moreover, because $\hat q_k$ is a binomial proportion,
%$
%\E|\hat q_k-q_k|^3 = O(n^{-2}),
%$
%so that $\E(\mathcal R_n)=O(n^{-2})$. Taking expectations and using $\E(\hat q_k-q_k)=0$, we obtain
%\[
%\E\{g_k(\hat q_k)\}
%=
%g_k(q_k)
%+\frac12 g_k''(q_k)\E\bigl[(\hat q_k-q_k)^2\bigr]
%+O(n^{-2}).
%\]
%Since $g_k(q_k)=\theta$ and
%$
%\E\bigl[(\hat q_k-q_k)^2\bigr]
%=
%\V(\hat q_k)
%=
%\frac{q_k(1-q_k)}{n},
%$
%it follows that
%\[
%\E(\hthet_{\mathrm{DUST\text{-}MNS}})
%=
%\theta
%+\frac12 g_k''(q_k)\frac{q_k(1-q_k)}{n}
%+O(n^{-2}),
%\]
%which gives \eqref{eq:bias}.
%\end{proof}
%

\begin{remark}\label{rem:bias}
Because $g_k(x)=1-(1-x)^{1/k}$ is strictly convex for $k>1$, Jensen's inequality implies
\[
\E\{\hthet_\mathrm{DUST\text{-}MNS}\}=\E\{g_k(\hat q_k)\}\ge  g_k(\E[\hat q_k])=g_k(q_k)=\theta.
\]
Hence $\hthet_\mathrm{DUST\text{-}MNS}$ is positively biased in finite samples. From \eqref{eq:bias}, the leading bias is of order $O(n^{-1})$ when $k$ is fixed, so increasing the sample size  $n$ reduces the bias. The set size $k$ also affects the bias through the curvature of $g_k$. For fixed $n$, increasing $k$ generally increases the magnitude of the transformation bias. Thus there is a practical tradeoff. Larger $k$ strengthens the tail-targeting effect of maxima nomination, but for fixed $n$ it also increases finite-sample bias and, in practice, may lead to greater ranking error. This provides further justification for treating $k$ as a small fixed design parameter and controlling precision primarily through $n$.
\end{remark}

\begin{remark}\label{rem:implementation}
In practice, $p_{[k:k],i}$ is not directly observed. The ideal exceedance indicator
$\mathbf{1}(p_{[k:k],i}>c)$ is replaced by an empirical version based on the observed binary outcomes  of a random sample of size $m_i$ in the selected area. If $X_{[k:k],i}\sim\mathrm{Binomial}(m_i,p_{[k:k],i})$, a natural estimator is
\[
\hat Z_i=\mathbf{1}\!\left(\frac{X_{[k:k],i}}{m_i}>c\right).
\]
By the law of large numbers, $X_{[k:k],i}/m_i \to p_{[k:k],i}$ almost surely as $m_i\to\infty$, so $\hat Z_i$ converges to $\mathbf{1}(p_{[k:k],i}>c)$. 
%For finite $m_i$, one may instead use the smoothed indicator
%\begin{equation}\label{eq:smoothed}
%\hat{Z}_i^{(m_i)} = 1 - \sum_{j=0}^{\lfloor m_i c\rfloor}\binom{m_i}{j}
%\left(\frac{X_{[k:k],i}}{m_i}\right)^j\left(1-\frac{X_{[k:k],i}}{m_i}\right)^{m_i-j},
%\end{equation}
%which satisfies
%\[
%\E\!\left[\hat{Z}_i^{(m_i)}\mid p_{[k:k],i}\right]
%=
%\Pr\!\left(\frac{X_{[k:k],i}}{m_i}>c \,\middle|\, p_{[k:k],i}\right),
%\]
%and approaches $\mathbf{1}(p_{[k:k],i}>c)$ as $m_i\to\infty$.
\end{remark}

We construct  three estimators  corresponding  to increasing levels of design sophistication.

%======================================================================
\subsection{SRS estimator}
%======================================================================
Under SRS, $n$ areas are selected at random without replacement from
$\{1,\ldots,N\}$. For a selected area $i$, let
$
X_i=\sum_{j=1}^{m_i} Y_{ij},
$
so that, conditional on $p_i$,
$
X_i \mid p_i \sim \mathrm{Binomial}(m_i, p_i).
$
Define the exceedance indicator
$
Z_i=\mathbf{1}\!(\tfrac{X_i}{m_i}>c).
$
The SRS estimator is then
\begin{equation}\label{eq:srs-est}
\bthet_{\mathrm{SRS}}=\frac{1}{n}\sum_{i=1}^n Z_i
=\frac{1}{n}\sum_{i=1}^n\mathbf{1}\!\left(\frac{X_i}{m_i}>c\right).
\end{equation}

For finite $m_i$, $\bthet_{\mathrm{SRS}}$ targets $\Pr(\frac{X_i}{m_i}>c)$; when $m_i$ is large, this is a close approximation to $\theta=\Pr(p_i>c)$. Its variance, under the working spatial model,  follows from the standard decomposition
\begin{equation}\label{eq:var-srs}
\V(\bthet_{\mathrm{SRS}}) = \frac{\theta(1-\theta)}{n} +
\frac{\sigma^2_\theta}{n^2}\sum_{i=1}^n\sum_{j\ne i}\eta_0^{l_{ij}},
\end{equation}
where $\sigma^2_\theta=\theta(1-\theta)$ since the exceedance indicator has Bernoulli
variance, and the spatial covariance term captures the inter-area dependence in exceedance
events.

\subsection{DUST-SRS estimator}

Under DUST-SRS, $n$ areas are selected via pps-DUST.  The estimator is identical in
form to (\ref{eq:srs-est}) but applied to DUST-selected areas.  Because DUST spreads areas
to minimize spatial overlap, the covariance term vanishes to a good approximation, giving
\begin{equation}\label{eq:var-dust-srs}
\V(\tthet_{\mathrm{DUST\text{-}SRS}}) \approx \frac{\theta(1-\theta)}{n}.
\end{equation}

\begin{remark}[Finite-population correction]\label{rem:fpc}
The variance formulae (\ref{eq:var-srs}) and (\ref{eq:var-dust-srs}) are derived under
a superpopulation model in which the $N$ areal units are treated as an effectively
infinite population relative to the sample size $n$.  When the sampling fraction
$f=n/N$ is non-negligible, both variances should be multiplied by the finite-population
correction factor $(1-f)$.  Specifically,
\begin{equation}\label{eq:var-srs-fpc}
\V_{\mathrm{FPC}}(\bthet_{\mathrm{SRS}})
= (1-f)\left[\frac{\theta(1-\theta)}{n} +
\frac{\theta(1-\theta)}{n^2}\sum_{i}\sum_{j\ne i}\eta_0^{l_{ij}}\right],
\end{equation}
\begin{equation}\label{eq:var-dust-srs-fpc}
\V_{\mathrm{FPC}}(\tthet_{\mathrm{DUST\text{-}SRS}})
\approx (1-f)\,\frac{\theta(1-\theta)}{n}.
\end{equation}
All relative efficiency results  are unchanged by the correction because the factor $(1-f)$ appears in both numerator
and denominator of every efficiency ratio and cancels exactly.  The correction is
therefore relevant for planning the absolute precision of the estimator but does not
affect design comparisons.  In the county stroke application, for sample sizes considered  the correction is negligible in practice.
\end{remark}

\subsection{DUST-MNS estimator}

Under DUST-MNS with parameters $(n,k)$, the estimator is $\hthet_{\mathrm{DUST\text{-}MNS}}$ from
(\ref{eq:thetahat}).  Its variance, obtained by the delta method from
Theorem~\ref{thm:unbiased}, is:
\begin{equation}\label{eq:var-mns}
\V(\hthet_{\mathrm{DUST\text{-}MNS}}) \approx [g_k'(q_k)]^2\cdot\frac{q_k(1-q_k)}{n},
\end{equation}
where $g_k'(q_k)=\frac{1}{k}(1-q_k)^{1/k-1}$ and $q_k=1-(1-\theta)^k$.
Substituting:
\begin{equation}\label{eq:var-mns-explicit}
\V(\hthet_{\mathrm{DUST\text{-}MNS}}) = \frac{[1-(1-\theta)^k](1-\theta)^{2-k}}{k^2\,n}.
\end{equation}
Note that  the variance (\ref{eq:var-mns-explicit}) decreases in $n$ at rate $1/n$ for any fixed $k$, so
precision is controlled entirely by the sample size $n$. Also, the set size $k$
enters through the factor $[1-(1-\theta)^k](1-\theta)^{2-k}/k^2$, which as we show later decreases in
$k$ for $\theta<\theta^\star(k)$, so larger sets improve
efficiency up to the critical threshold. 

%======================================================================
\section{Efficiency  Comparison of Proposed Estimators}\label{sec:effcomp}
%======================================================================

\begin{theorem}\label{thm:chain}
Under the spatial autocorrelation model (\ref{eq:autocorr}) with $\eta_0>0$, and with
{\rm{DUST-MNS}}  design with  parameters $(n,k)$,
\begin{enumerate}[(i)]
\item $\V(\bthet_{\mathrm{SRS}})\ge\V(\tthet_{\mathrm{DUST\text{-}SRS}})$, with
equality iff\/ $\eta_0=0$.
\item $\V(\tthet_{\mathrm{DUST\text{-}SRS}})\ge\V(\hthet_{\mathrm{DUST\text{-}MNS}})$ if and only
if $\theta\le\theta^\star(k)$, where $\theta^\star(k)=1-u_k^\star$ and
$u_k^\star\in(0,1)$ is the unique solution of
\begin{equation}\label{eq:thetastar}
k^2 u^{\,k-1} = 1+u+\cdots+u^{k-1}.
\end{equation}
\item Combining (i) and (ii), for $\theta\le\theta^\star(k)$,
\begin{equation}\label{eq:chain}
\V(\bthet_{\mathrm{SRS}})\;\ge\;\V(\tthet_{\mathrm{DUST\text{-}SRS}})\;\ge\;
\V(\hthet_{\mathrm{DUST\text{-}MNS}}).
\end{equation}
\end{enumerate}
\end{theorem}

\begin{proof}
Part (i).  By comparing (\ref{eq:var-srs}) and (\ref{eq:var-dust-srs}),
the difference is $\frac{\sigma^2_\theta}{n^2}\sum_i\sum_{j\ne i}\eta_0^{l_{ij}}>0$
whenever $\eta_0>0$. Part (ii).  We need $\V(\tthet)\ge\V(\hthet_\mathrm{DUST\text{-}MNS})$, i.e.,
$\theta(1-\theta)/n \ge [1-(1-\theta)^k](1-\theta)^{2-k}/(k^2 n)$.
This simplifies  to
$k^2\theta(1-\theta)\ge [1-(1-\theta)^k](1-\theta)^{2-k}$.
Writing $u=1-\theta\in(0,1)$ and multiplying through by $u^{k-2}/(1-u)$ gives
$
k^2 u^{k-1}\ge 1+u+\cdots+u^{k-1}.
$
Now define
\[
\phi_k(u)=\frac{1+u+\cdots+u^{k-1}}{u^{k-1}}
=1+u^{-1}+u^{-2}+\cdots+u^{-(k-1)}, \qquad u\in(0,1).
\]
Then $\V(\tthet)\ge\V(\hthet_\mathrm{DUST\text{-}MNS})$ is equivalent to
$
k^2 \ge \phi_k(u).
$
The function $\phi_k(u)$ is continuous and strictly decreasing on $(0,1)$, with
$
\lim_{u\downarrow 0}\phi_k(u)=\infty,
$
and
$
\lim_{u\uparrow 1}\phi_k(u)=k.
$
Since $k^2>k$ for every $k\ge 2$, there exists a unique $u_k^\star\in(0,1)$ such that
$
\phi_k(u_k^\star)=k^2.
$
Therefore
\[
\V(\tthet_{\mathrm{DUST-SRS}})\ge \V(\hthet_{\mathrm{DUST\text{-}MNS}})
\quad\Longleftrightarrow\quad
u\ge u_k^\star
\quad\Longleftrightarrow\quad
\theta\le 1-u_k^\star.
\]
Defining
$
\theta^\star(k)=1-u_k^\star,
$
we conclude that DUST-MNS dominates DUST-SRS whenever
$
\theta\le \theta^\star(k).
$
Part (iii) follows from (i) and (ii).
\end{proof}

\begin{remark}\label{rem:thetastar}
The critical threshold depends only on the set size $k$, not on the sample size $n$.
This is a key structural property of the  design as  the analyst can increase
$n$ to improve precision without affecting the regime in which the design is efficient.
The exact threshold is obtained numerically from (\ref{eq:thetastar}).  For example,
\[
\theta^\star(2)=0.667,\quad
\theta^\star(3)=0.578,\quad
\theta^\star(4)=0.514,\quad
\theta^\star(5)=0.464,\quad
\theta^\star(6)=0.425,\quad
\theta^\star(10)=0.322.
\]
For areal surveillance applications, $k\in\{3,4,5\}$ gives $\theta^\star(k)\in(0.46,0.58)$,
which comfortably covers the typical exceedance regime. 
\end{remark}

Here we obtain a lower bound for the relative efficiency of  DUST-SRS estimator  versus the SRS-based estimator.

\begin{theorem}\label{thm:lambda}
Assume the pairwise neighborhood lags among sampled areas are represented by a random variable
$L\sim g(\cdot)$ with mean $\bar l$. Then
\begin{equation}\label{eq:lambda}
\mathrm{RE}(\tthet_{\mathrm{DUST\text{-}SRS}},\bthet_{\mathrm{SRS}})
=
1+(n-1)\E\!\left(\eta_0^L\right)
\;\ge\;
1+\Lambda,
\end{equation}
where
\begin{equation}\label{eq:Lambda}
\Lambda=(n-1)\eta_0^{\bar l}.
\end{equation}
\end{theorem}

\begin{proof}
From \eqref{eq:var-srs} and \eqref{eq:var-dust-srs},
\[
\mathrm{RE}(\tthet_{\mathrm{DUST\text{-}SRS}},\bthet_{\mathrm{SRS}})
=
\frac{\V(\bthet_{\mathrm{SRS}})}{\V(\tthet_{\mathrm{DUST\text{-}SRS}})}
=
1+\frac{1}{n}\sum_{i\ne j}\eta_0^{\,l_{ij}}.
\]
If the off-diagonal lags are summarized by a generic random variable $L$, then
\[
\frac{1}{n(n-1)}\sum_{i\ne j}\eta_0^{\,l_{ij}}
=
\E(\eta_0^L),
\]
and therefore
\[
\mathrm{RE}(\tthet_{\mathrm{DUST\text{-}SRS}},\bthet_{\mathrm{SRS}})
=
1+(n-1)\E(\eta_0^L).
\]
Since $x\mapsto \eta_0^x$ is convex on $[1,\infty)$ for $0<\eta_0<1$, Jensen's inequality gives
\[
\E(\eta_0^L)\ge \eta_0^{\E(L)}=\eta_0^{\bar l}.
\]
Substituting this bound yields
$
\mathrm{RE}(\tthet_{\mathrm{DUST\text{-}SRS}},\bthet_{\mathrm{SRS}})
\ge
1+(n-1)\eta_0^{\bar l},
$
which proves the result.
\end{proof}

\begin{remark}
The SRS variance (\ref{eq:var-srs}) involves the neighbourhood lag structure through
$\sum_{ij}\eta_0^{l_{ij}}$.  For planning, we model the random lag $L$ between two
randomly selected areas with PMF $g(l)=2N_l/[N(N-1)]$, where $N_l$ is the number of
area pairs at lag $l$.  The zero-truncated binomial distribution ZTBin$(\nu,\rho)$
approximates $g(\cdot)$, with mean
\begin{equation}\label{eq:ztbin-mean}
\bar{\bar{l}}(\rho;\nu) = \sum_{l=1}^\nu
l\,\frac{\binom{\nu}{l}\rho^l(1-\rho)^{\nu-l}}{1-(1-\rho)^\nu}.
\end{equation}
Only the mean lag $\bar{l}$ is required to evaluate the minimum efficiency bound
$1+\Lambda_\theta$ (Theorem~\ref{thm:lambda}).

\end{remark}

Now, we provide the relative efficiency  of DUST-MNS over DUST-SRS estimator.  Define $\Delta_\theta = 1 - \V(\hthet_\mathrm{DUST\text{-}MNS})/\V(\tthet_{\rm DUST\text{-}SRS})$, so
that $\mathrm{RE}(\tthet_{\rm DUST\text{-}SRS},\hthet_\mathrm{DUST\text{-}MNS}) = 1/(1-\Delta_\theta)$.
From (\ref{eq:var-mns-explicit}) and (\ref{eq:var-dust-srs}):
\begin{equation}\label{eq:Delta}
\Delta_\theta = 1 - \frac{[1-(1-\theta)^k](1-\theta)^{2-k}}{k^2\,\theta(1-\theta)} \quad \text{and}\quad \mathrm{RE}=\frac{k^2\theta(1-\theta)^{k-1}}{1-(1-\theta)^k}.
\end{equation}
Note that $\Delta_\theta$ and  RE depend only on $k$ and $\theta$, not on $n$.  The following result provides  some useful properties of $\Delta_\theta$.

\begin{theorem}\label{thm:delta-props}
Let $\Delta_\theta$ be defined by (\ref{eq:Delta}).
Then:
\begin{enumerate}[(a)]
\item $\Delta_\theta>0$  (equivalently, DUST-MNS dominates DUST-SRS)  if and only if $\theta<\theta^\star(k)$, where $\theta^\star(k)$ is defined by the equality case $\V(\hthet_\mathrm{DUST\text{-}MNS})=\V(\tthet_{\rm DUST\text{-}SRS})$.
\item $\Delta_\theta$ is strictly decreasing in $\theta\in(0,1)$.
\item
$
\lim_{\theta\to 0}\Delta_\theta = 1-\frac{1}{k}.
$
\item In general, $\Delta_\theta(\theta)\neq \Delta_\theta(1-\theta)$.
\end{enumerate}
\end{theorem}

\begin{proof}
Part (a) follows directly from the definition
 of $\Delta_\theta$. 
It is easy to see that  $\Delta_\theta>0$ holds exactly when $\V(\hthet_\mathrm{DUST\text{-}MNS})<\V(\tthet_{\rm DUST\text{-}SRS})$, which is equivalent to $\theta<\theta^\star(k)$. For part (b), write $u=1-\theta$. Then
\[
\Delta_\theta
=
1-\frac{1-u^k}{k^2(1-u)}u^{1-k}
=
1-\frac{1+u+\cdots+u^{k-1}}{k^2u^{k-1}}
=
1-\frac{1}{k^2}\sum_{j=0}^{k-1}(1-\theta)^{-j}.
\]
Differentiating gives
\[
\frac{d}{d\theta}\Delta_\theta
=
-\frac{1}{k^2}\sum_{j=1}^{k-1} j(1-\theta)^{-j-1}<0,
\]
so $\Delta_\theta$ is strictly decreasing on $(0,1)$.
For part (c), as $\theta\to 0$,
\[
1-(1-\theta)^k \sim k\theta,
\qquad
(1-\theta)^{2-k}\to 1,
\qquad
\theta(1-\theta)\sim\theta.
\]
Hence
\[
\frac{[1-(1-\theta)^k](1-\theta)^{2-k}}{k^2\theta(1-\theta)}
\to \frac{1}{k},
\]
and therefore
$
\Delta_\theta\to 1-\frac{1}{k}.
$
Part (d) is immediate from the explicit formula, which is not symmetric in $\theta$ and $1-\theta$.
\end{proof}

%%%%%%%%%%%%%%%%%%%%%%%%%%%%%%%%
%
%
%%%%%%%%%%%%%%%%%%%%%%%%%%%%%%%%

\subsection{Beta model efficiency evaluation}\label{sec:efficiency}

To obtain explicit efficiency expressions, we specialize to the case
\[
p_i\sim \mathrm{Beta}(\alpha,\beta).
\]
Then the exceedance probability is
$
\theta=\Pr(p_i>c)=1-I_c(\alpha,\beta),
$
where $I_c(\alpha,\beta)$ denotes the regularized incomplete beta function. Under the working model,
\[
\V(\tthet_{\mathrm{DUST\text{-}SRS}})=\frac{I_c(\alpha,\beta)(1-I_c(\alpha,\beta))  }{n}.
\]
Moreover, the lower bound on the relative efficiency of DUST-SRS relative to SRS is
\begin{equation}\label{eq:Lambda-beta}
1+\Lambda(\eta_0,n,\bar l)=1+(n-1)\eta_0^{\bar l},
\end{equation}
which does not depend on $(\alpha,\beta)$.
For DUST-MNS relative to DUST-SRS we have 
\begin{equation}\label{eq:RE-beta-closed}
\mathrm{RE}(\hthet_{\mathrm{DUST\text{-}MNS}},\tthet_{\mathrm{DUST\text{-}SRS}})
=
\frac{\V(\tthet_{\mathrm{DUST\text{-}SRS}})}{\V(\hthet_{\mathrm{DUST\text{-}MNS}})}
=
\frac{k^2\{1-I_c(\alpha,\beta)\}\,I_c(\alpha,\beta)^{k-1}}
{1-I_c(\alpha,\beta)^k}.
\end{equation}
Thus, once $c$, $\alpha$, and $\beta$ are specified, the relative efficiency is determined explicitly.

Tables~\ref{tab:Lambda}--\ref{tab:bias} summarize the resulting numerical evaluations. Table~\ref{tab:Lambda} reports the lower bound $1+\Lambda(\eta_0,n,\bar l)$ for DUST-SRS relative to SRS. Table~\ref{tab:Delta} reports $\mathrm{RE}(\hthet_{\mathrm{DUST\text{-}MNS}},\tthet_{\mathrm{DUST\text{-}SRS}})$ for selected values of $\theta$ and $k$; entries with $\theta>\theta^\star(k)$ are marked by $\dagger$, indicating that DUST-MNS is no longer more efficient than DUST-SRS. Table~\ref{tab:thetastar} gives the critical threshold $\theta^\star(k)$. Finally, Table~\ref{tab:bias} reports the exact finite-sample bias of the calibrated MNS estimator for $n=10,20$ and $k=2,3,4,5$.

Several patterns are clear. The gain of DUST-SRS over SRS increases with $n$ and with the strength of spatial dependence, as reflected in $\eta_0^{\bar l}$. For DUST-MNS relative to DUST-SRS, efficiency gains are largest when $\theta$ is small and diminish as $\theta$ approaches $\theta^\star(k)$. The bias table shows a different pattern: for fixed $k$, the bias decreases as $n$ increases, while for fixed $n$, it increases with $k$. Thus, although larger set sizes can improve variance performance in the rare-exceedance regime, they also induce greater finite-sample bias.

\begin{table}[h!]
\caption{Lower bound on the relative efficiency of DUST-SRS relative to SRS,
$
1+\Lambda(\eta_0,n,\bar l)=1+(n-1)\eta_0^{\bar l}.
$
This bound does not depend on $\theta$ or on $(\alpha,\beta)$.}
\label{tab:Lambda}
\begin{center}
\begin{tabular}{cc|cccccc}
\toprule
 & & \multicolumn{3}{c}{$\bar{l}=1.6$, $n=10$} & \multicolumn{3}{c}{$\bar{l}=2.4$, $n=20$}\\
\cmidrule(lr){3-5}\cmidrule(lr){6-8}
 & & $\eta_0=0.2$ & $\eta_0=0.5$ & $\eta_0=0.8$ & $\eta_0=0.2$ & $\eta_0=0.5$ & $\eta_0=0.8$\\
\midrule
\multicolumn{2}{c|}{All $\theta$, all $(\alpha,\beta)$} & 1.685 & 3.969& 7.298& 1.399 & 4.599& 12.122\\
\bottomrule
\end{tabular}
\end{center}
\end{table}

\begin{table}[h!]
\caption{Relative efficiency of DUST-MNS relative to DUST-SRS,
$\mathrm{RE}(\hthet_{\mathrm{DUST\text{-}MNS}},\tthet_{\mathrm{DUST\text{-}SRS}})$ 
for selected values of $\theta$ and set size $k$. Entries with $\theta>\theta^\star(k)$ are marked with $\dagger$.}
\label{tab:Delta}
\begin{center}
\begin{tabular}{c|cccccc}
\toprule
$\theta$ & $k=2$ & $k=3$ & $k=4$ & $k=5$ & $k=6$ & $k=10$\\
\midrule
0.05 & 1.949 & 2.848 & 3.698 & 4.501 & 5.258 & 7.853\\
0.10 & 1.895 & 2.690 & 3.392 & 4.005 & 4.537 & 5.948\\
0.15 & 1.838 & 2.528 & 3.084 & 3.519 & 3.847 & 4.326\\
0.20 & 1.778 & 2.361 & 2.775 & 3.046 & 3.198 & 3.007\\
0.25 & 1.714 & 2.189 & 2.469 & 2.593 & 2.598 & 1.989\\
0.30 & 1.647 & 2.014 & 2.167 & 2.165 & 2.057 & 1.246\\
0.34 & 1.590 & 1.871 & 1.930 & 1.844 & 1.671 & $\dagger$\\
0.35 & 1.576 & 1.835 & 1.872 & 1.767 & 1.581 & $\dagger$\\
0.40 & 1.500 & 1.653 & 1.588 & 1.405 & 1.175 & $\dagger$\\
0.42 & 1.468 & 1.580 & 1.478 & 1.272 & 1.032 & $\dagger$\\
\bottomrule
\end{tabular}
\end{center}
\end{table}

\begin{table}[h!]
\caption{Critical threshold $\theta^\star(k)$. For $\theta<\theta^\star(k)$, DUST-MNS is more efficient than DUST-SRS; for $\theta\ge \theta^\star(k)$, DUST-SRS is at least as efficient. The threshold depends only on the set size $k$, not on the sample size $n$.}
\label{tab:thetastar}
\begin{center}
\begin{tabular}{c|cccccccc}
\toprule
$k$ & 2 & 3 & 4 & 5 & 6 & 7 & 8 & 10\\
\midrule
$\theta^\star(k)$ & 0.6667 & 0.5785 & 0.5140 & 0.4643 & 0.4247 & 0.3921 & 0.3649 & 0.3215\\
\bottomrule
\end{tabular}
\end{center}
\end{table}

\begin{table}[h!]
\caption{Exact finite-sample bias $\E[\hthet_\mathrm{DUST\text{-}MNS}]-\theta$ for the calibrated MNS estimator, for sample sizes $n=10$ and $n=20$ and set sizes $k=2,3,4,5$. }
\label{tab:bias}
\begin{center}
\begin{tabular}{c|cccc|cccc}
\toprule
& \multicolumn{4}{c|}{$n=10$} & \multicolumn{4}{c}{$n=20$} \\
\cmidrule(lr){2-5}\cmidrule(lr){6-9}
$\theta$ & $k=2$ & $k=3$ & $k=4$ & $k=5$ & $k=2$ & $k=3$ & $k=4$ & $k=5$ \\
\midrule
0.10 & 0.0028 & 0.0040 & 0.0049 & 0.0057 & 0.0014 & 0.0019 & 0.0023 & 0.0026 \\
0.20 & 0.0061 & 0.0099 & 0.0149 & 0.0251 & 0.0029 & 0.0045 & 0.0059 & 0.0076 \\
0.30 & 0.0104 & 0.0219 & 0.0492 & 0.1024 & 0.0048 & 0.0083 & 0.0135 & 0.0273 \\
0.40 & 0.0169 & 0.0510 & 0.1265 & 0.2305 & 0.0072 & 0.0158 & 0.0424 & 0.1074 \\
\bottomrule
\end{tabular}
\end{center}
\end{table}

%======================================================================
\section{Imperfect Ranking and Inference}\label{sec:imperf}
%======================================================================

In practice, the maximum area in each set is identified using an auxiliary variable $a_i$
rather than the true $p_i$. Let
\[
q_k^{\,e}=\sum_{r=1}^k \nu_{rk}\,\Pr(p_{r:k}>c)
\]
denote the exceedance probability of the nominated unit under a doubly stochastic misranking matrix $\boldsymbol{\nu}=(\nu_{rs})$ in the sense of \citet{BohnWolfe1994}. Related ideas for handling uncertain or probabilistic rankings also appear in the judgment post-stratification literature; see \citet{MacEachernStasnyWolfe2004}.
Then $q_k^{\,e}\le q_k$, with equality only under perfect ranking. A key point is that the perfect-ranking inverse
$
\theta=1-(1-q_k)^{1/k}
$
no longer applies under imperfect ranking. Thus the estimator
\[
1-\left(1-\frac{R_n^{\,e}}{n}\right)^{1/k}
\]
is generally biased for $\theta$ when ranking is imperfect. A coherent analysis must instead calibrate through the imperfect-ranking map
\[
h_\nu(\theta)=q_k^{\,e}(\theta)
=\sum_{r=1}^k \nu_{rk}\,\Pr(p_{r:k}>c),
\]
and define
\[
\hat\theta_{\mathrm{DUST\text{-}MNS}}^{\,e}=h_\nu^{-1}(R_n^{\,e}/n),
\]
provided that $h_\nu$ is strictly increasing.

\subsection{\texorpdfstring{A one-parameter working model based on Kendall's $\tau$}{A one-parameter working model based on Kendall's tau}}\label{ssec:tau}

For sensitivity analysis, it is convenient to work with the interpolation
\begin{equation}\label{eq:qeff}
q_k^{\,e}(\tau)=\tau^2 q_k + (1-\tau^2)\theta,
\qquad \tau\in[0,1],
\end{equation}
which connects perfect ranking ($\tau=1$, so $q_k^{\,e}=q_k$) and random nomination
($\tau=0$, so $q_k^{\,e}=\theta$). This is used here as a pragmatic working model rather than a structural consequence of the ranking mechanism.

Under \eqref{eq:qeff}, the imperfect-ranking estimator is obtained by solving
\[
q_k^{\,e}(\tau;\theta)=\frac{R_n^{\,e}}{n}
\]
for $\theta$, typically by one-dimensional numerical root finding. The corresponding delta-method variance is then based on the derivative of the inverse map $h_\tau^{-1}$,
\[
\V(\hat\theta_{\mathrm{DUST\text{-}MNS}}^{\,e})
\approx
\left[\frac{d\theta}{dq_k^{\,e}}\right]^2
\frac{q_k^{\,e}(1-q_k^{\,e})}{n},
\]
where
\[
\frac{dq_k^{\,e}}{d\theta}
=
\tau^2\frac{dq_k}{d\theta}+(1-\tau^2),
\qquad
\frac{dq_k}{d\theta}=k(1-\theta)^{k-1}.
\]
Thus imperfect ranking alters both the effective exceedance probability and the calibration slope. As $\tau$ decreases, $q_k^{\,e}$ moves from the perfect-ranking value $q_k$ toward the random-nomination baseline $\theta$, and the advantage of maxima nomination is correspondingly reduced.

For the county stroke application, the empirical Kendall correlation between stroke prevalence and diabetes prevalence is high, so the imperfect-ranking estimator is expected to remain close to the perfect-ranking benchmark. In the numerical section, we therefore treat \eqref{eq:qeff} as a sensitivity model and examine how the estimator behaves over a range of plausible $\tau$ values, with particular attention to the empirical value observed in the data.

\subsection{Variance estimation and confidence intervals}\label{sec:inference}

The relevant asymptotic regime for the present design holds the set size $k$ fixed and lets the number of sets $n$ increase. Under the working approximation that the selected maxima are independent across sets, we have
$
R_n\sim \mathrm{Binomial}(n,q_k).
$
Hence
\[
\sqrt{n}\,(\hat q_k-q_k)\xrightarrow{d}N\!\bigl(0,q_k(1-q_k)\bigr),
\qquad
\hat q_k=\frac{R_n}{n}.
\]
Since
$
\hat\theta_{\mathrm{DUST\text{-}MNS}}=g_k(\hat q_k),
$
and 
$
g_k(x)=1-(1-x)^{1/k},
$
the delta method yields
\begin{equation}\label{eq:sigma_mns}
\sqrt{n}\,\bigl(\hat\theta_{\mathrm{DUST\text{-}MNS}}-\theta\bigr)
\xrightarrow{d}
N\!\left(0,\,[g_k'(q_k)]^2q_k(1-q_k)\right).
\end{equation}
%where
%\begin{equation}\label{eq:sigma_mns}
%\V(\hat\theta_{\mathrm{DUST\text{-}MNS}})
%\approx
%[g_k'(q_k)]^2\frac{q_k(1-q_k)}{n}
%=
%\frac{[1-(1-\theta)^k](1-\theta)^{2-k}}{k^2\,n}.
%\end{equation}
%
For moderate $n$, the normal approximation takes the form
\[
\hat\theta_{\mathrm{DUST\text{-}MNS}} \approx N\!\left(
\theta + \frac12 g_k''(q_k)\frac{q_k(1-q_k)}{n},
\; [g_k'(q_k)]^2\frac{q_k(1-q_k)}{n}
\right),
\]
with
\[
g_k''(q_k)=\frac{k-1}{k^2}(1-q_k)^{1/k-2}.
\]
Thus the leading finite-sample bias is of order $O(n^{-1})$, while the variance is of order $O(n^{-1})$ for fixed $k$. A plug-in estimator of $\V(\hat\theta_{\mathrm{DUST\text{-}MNS}})$ is obtained from \eqref{eq:sigma_mns} by replacing $\theta$ with $\hat\theta_{\mathrm{DUST\text{-}MNS}}$and an  approximate $(1-\alpha)$ confidence interval for $\theta$ is therefore
\begin{equation}\label{eq:CI}
\hat\theta_{\mathrm{DUST\text{-}MNS}}
-
\widehat{\mathrm{Bias}}(\hat\theta_{\mathrm{DUST\text{-}MNS}})
\;\pm\;
z_{\alpha/2}\sqrt{\widehat{\V}(\hat\theta_{\mathrm{DUST\text{-}MNS}})},
\end{equation}
where $\widehat{\mathrm{Bias}}(\hat\theta_{\mathrm{DUST\text{-}MNS}})$ is obtained by plug-in substitution in the leading bias formula \eqref{eq:bias}. For small or moderate $n$, a bootstrap interval may be preferable.

The leading bias term in \eqref{eq:bias} suggests the bias-corrected estimator
\begin{equation}\label{eq:bc_est}
\hat\theta_{\mathrm{BC}}
=
\hat\theta_{\mathrm{DUST\text{-}MNS}}
-
\widehat{\mathrm{Bias}}(\hat\theta_{\mathrm{DUST\text{-}MNS}}),
\end{equation}
where
\begin{equation}\label{eq:bias_plugin}
\widehat{\mathrm{Bias}}(\hat\theta_{\mathrm{DUST\text{-}MNS}})
=
\frac{(k-1)}{2k^2\,n}
(1-\hat q_k)^{1/k-2}\,\hat q_k(1-\hat q_k),
\qquad
\hat q_k=1-(1-\hat\theta_{\mathrm{DUST\text{-}MNS}})^k.
\end{equation}
Under the fixed-$k$, large-$n$ regime, this removes the $O(n^{-1})$ bias term and leaves a residual bias of order $O(n^{-2})$, while the variance remains unchanged to first order. An approximate $(1-\alpha)$ confidence interval based on the bias-corrected estimator is
\begin{equation}\label{eq:CI_bc}
\hat\theta_{\mathrm{BC}}
\;\pm\;
z_{\alpha/2}\sqrt{\widehat{\V}(\hat\theta_{\mathrm{DUST\text{-}MNS}})}.
\end{equation}

%======================================================================
\section{Estimating County-Level Stroke Prevalence in the United States}\label{sec:application}
%======================================================================

Stroke is a leading cause of death and long-term disability in the United States, and
identifying geographic concentrations of elevated burden is a central objective of
chronic-disease surveillance \citep{GreenlundEtAl2022,Hacker2024}. County-level chronic
disease mapping is routinely used to allocate prevention resources and guide geographically
targeted interventions
\citep{StulbergEtAl2024,CarlsonEtAl2023,ChenEtAl2026,BenavidezEtAl2024}. In this setting,
the inferential question is naturally threshold-based. A  public-health agency may wish to
estimate the proportion of counties whose stroke prevalence exceeds a high-burden
benchmark, rather than summarize the national mean prevalence. This is precisely the
estimand $\theta=\Pr(p_i>c)$ for which DUST-MNS is designed. The CDC PLACES platform
provides age-adjusted county-level chronic disease estimates for all U.S.\ counties and
county equivalents \citep{GreenlundEtAl2022}, yielding a large areal population with rich
spatial structure and a medically meaningful outcome.

This section uses the county data in two distinct but complementary ways. First, the
national county frame is treated as a fixed finite population.  We compute its observed exceedance proportion, evaluate its spatial and ranking
structure, and use that frame as the population from which repeated samples are drawn in
the Monte Carlo study. Second, the superpopulation development in earlier sections is used
as an organizing model for understanding why the design should work well in this setting.
Thus, the theory provides the efficiency benchmark and qualitative guidance, while the real
data analysis evaluates how the design behaves on the observed county population itself.

\subsection{Data and design inputs}\label{ssec:data}

Let $i=1,\ldots,N$ index U.S.\ counties (or county equivalents) with valid CDC PLACES
estimates. The study variable $p_i$ is the age-adjusted prevalence of stroke among adults
aged $\ge 18$ years, and the auxiliary ranking variable $a_i$ is the age-adjusted
prevalence of diagnosed diabetes in the same population, both taken from the same PLACES
county release. County population $M_i$, used for PPS weighting in the DUST step, is taken
from the population field distributed with the PLACES file and merged by county FIPS code.
County adjacency is derived from the U.S.\ Census county adjacency file and the county
cartographic boundary shapefile, which together define the contiguity graph used to compute
the neighborhood lag $l_{ij}$ and implement the DUST weights in \eqref{eq:dust}. All data
sources are publicly available, and the full workflow is implemented in the accompanying
script.

For the threshold-oriented surveillance objective, we set
\[
  c = Q_{0.90}(\{p_i: i=1,\ldots,N\}),
\]
the empirical 90th percentile of county stroke prevalence, so that $\theta=\Pr(p_i>c)$ is
the proportion of counties in the upper tail of the national distribution. After removing
counties with missing values, the analytic dataset contains $N=2921$ counties. The
empirical 90th percentile is $c=4.1\%$, giving a census exceedance proportion of
\[
\hat\theta = \frac{1}{N}\sum_{i=1}^N I(p_i>c) = 0.0866.
\]
This value is slightly below $0.10$ because the threshold is computed from the empirical
county distribution and the exceedance event is defined by the strict inequality $p_i>c$;
with a discrete finite set of county-level estimates, ties at or near the empirical 90th
percentile need not yield an exact upper-decile count. Table~\ref{tab:county_empirical_summary}
summarizes the main empirical features of the dataset. Two are especially favorable for
DUST-MNS. First, stroke prevalence exhibits strong spatial dependence, indicating that
spatial spreading should yield meaningful gains over ordinary random sampling. Second, the
auxiliary ranking variable is highly informative: the association between county stroke and
diabetes prevalence is strong on both linear and rank scales, suggesting that the
imperfect-ranking version of DUST-MNS should remain close to the ideal benchmark.

\begin{table}[htbp]
\centering
\caption{Empirical characteristics of the U.S.\ county-level stroke dataset.}
\label{tab:county_empirical_summary}
\begin{tabular}{lc}
\toprule
Quantity & Value \\
\midrule
Number of counties retained $(N)$ & 2921 \\
Threshold $c$ (90th percentile of stroke prevalence) & $4.1\%$ \\
Census exceedance proportion $\hat\theta$ & 0.0866 \\
Pearson correlation: stroke vs.\ diabetes & 0.9199 \\
Kendall $\tau$: stroke vs.\ diabetes & 0.7621 \\
Moran's $I$ for county stroke prevalence & 0.6141 \\
Mean graph lag $\bar l$ & 26.98 \\
\bottomrule
\end{tabular}
\end{table}

%These empirical values place the application in the rare-to-moderate exceedance regime,
%which is the setting in which maxima nomination is most attractive. The strong spatial
%dependence motivates the DUST spreading step, while the high concordance between stroke and
%diabetes prevalence makes diabetes prevalence a practically effective auxiliary variable for
%ranking counties within candidate sets.

\subsection{Empirical patterns and design diagnostics}\label{ssec:diagnostics}

Figures~\ref{fig:maps_panel} and~\ref{fig:analysis_panel} illustrate the two structural
features that motivate the design. Figure~\ref{fig:maps_panel} shows the continuous stroke
prevalence map together with the binary exceedance map. The prevalence choropleth reveals
pronounced spatial clustering of elevated burden in the Southeast, Appalachia, and the
lower Mississippi Delta, exactly the type of pattern for which DUST is intended to produce
better spatial coverage. The exceedance map confirms that high-burden counties occur in
contiguous clusters rather than as isolated units, so an ordinary random sample may either
over- or under-represent these regions.

Figure~\ref{fig:analysis_panel} shows the two quantities that drive the performance of
DUST-MNS in this application. The histogram confirms that the threshold $c=4.1\%$ lies in
the upper tail of the county stroke distribution, so the inferential target is genuinely a
tail probability rather than a mean. The scatter plot shows a strong positive association
between county stroke prevalence and diabetes prevalence. This supports the use of diabetes
prevalence as the auxiliary variable for within-set ranking, and suggests that the
imperfect-ranking version of DUST-MNS should closely approximate the perfect-ranking
benchmark.

\begin{figure}[htbp]
\centering
\safeincludegraphics[width=0.48\textwidth]{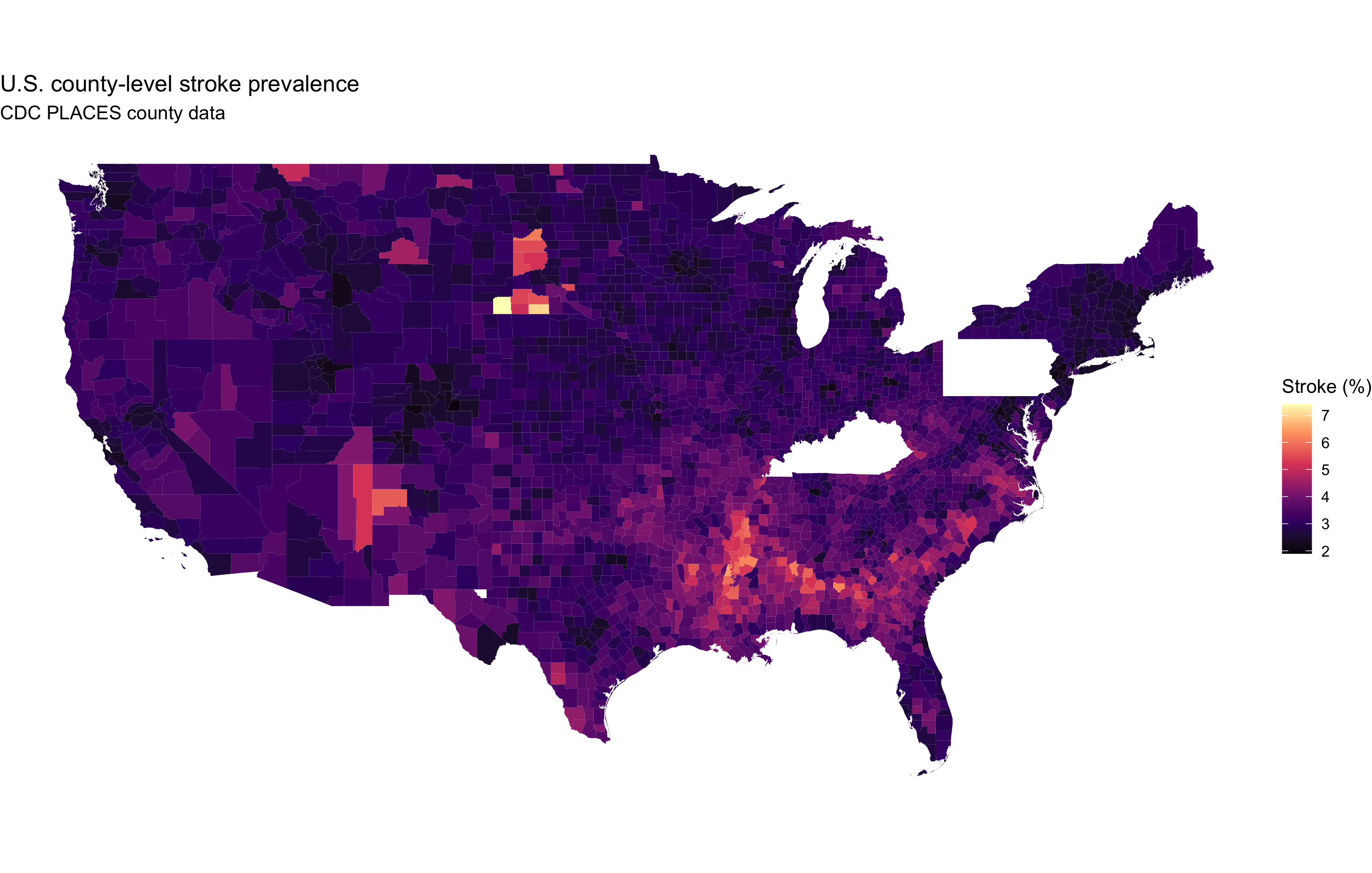}\hfill
\safeincludegraphics[width=0.48\textwidth]{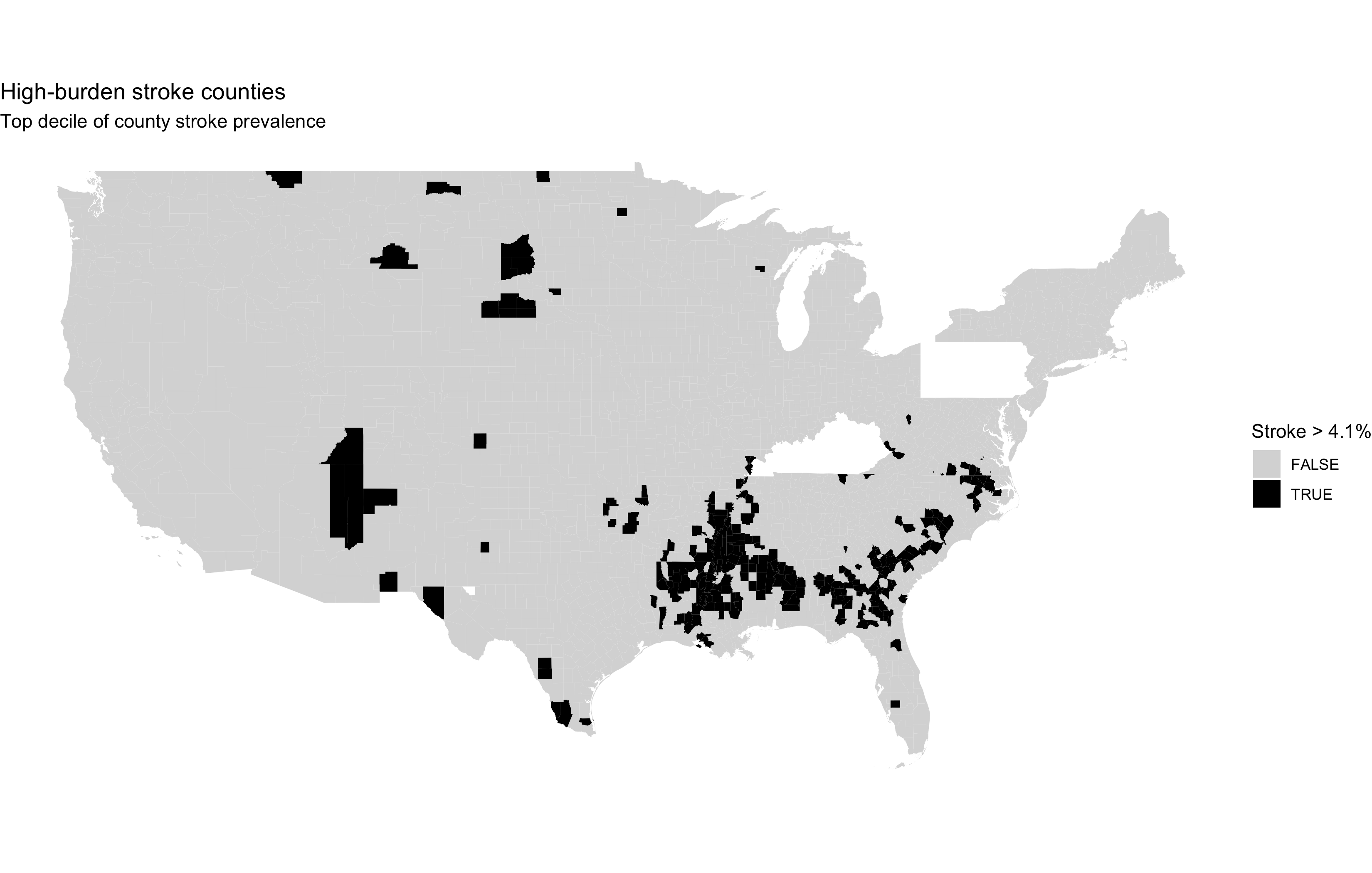}
\caption{Left: county-level stroke prevalence (age-adjusted, adults $\ge 18$ years) across
the retained U.S.\ counties. Elevated burden is concentrated in the Southeast, Appalachia,
and the lower Mississippi Delta. Right: counties whose stroke prevalence exceeds the
empirical 90th percentile threshold $c=4.1\%$ (shown in black). The clustering of
high-burden counties motivates the use of DUST to promote spatial spread.}
\label{fig:maps_panel}
\end{figure}

\begin{figure}[htbp]
\centering
\safeincludegraphics[width=0.48\textwidth]{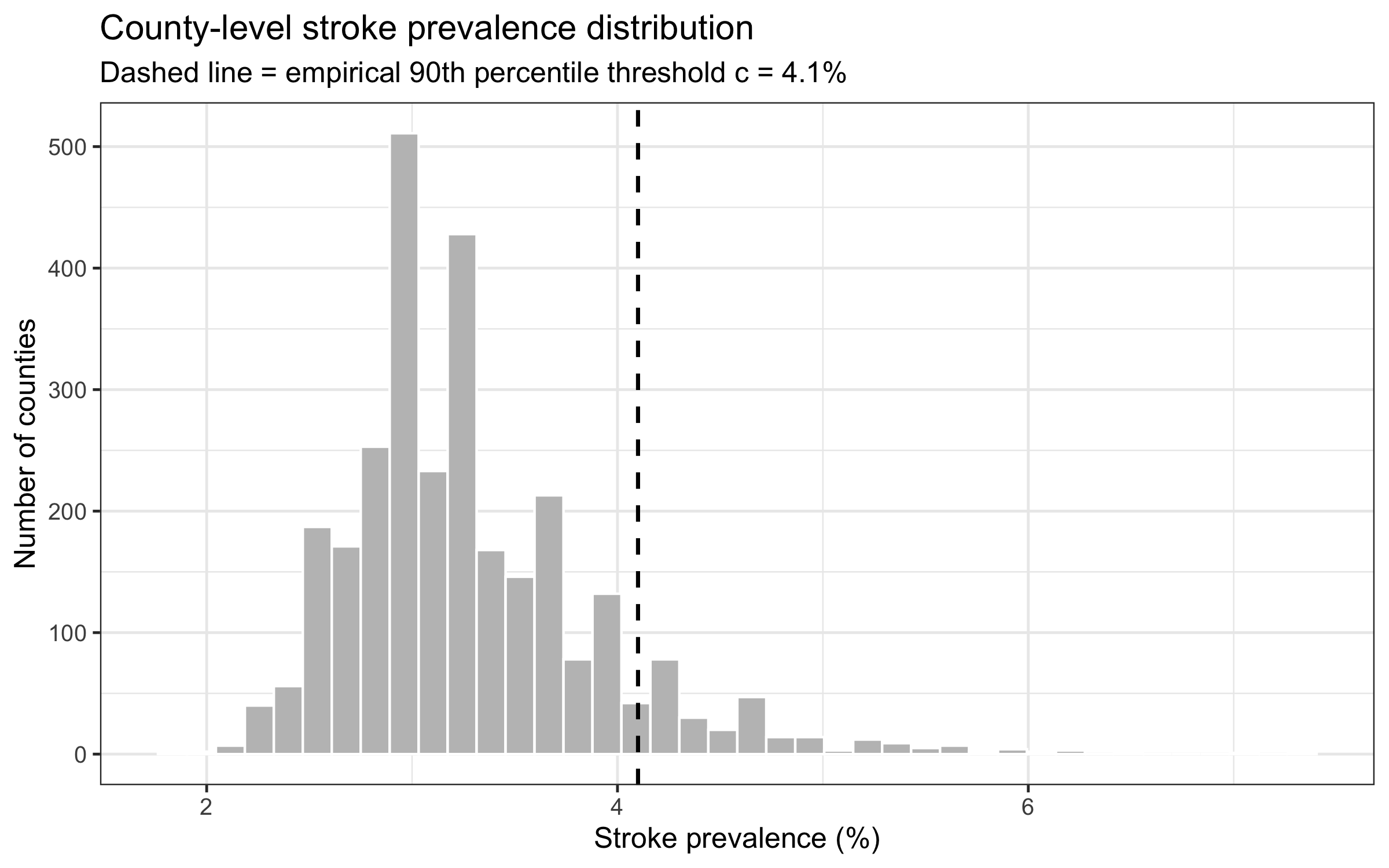}\hfill
\safeincludegraphics[width=0.48\textwidth]{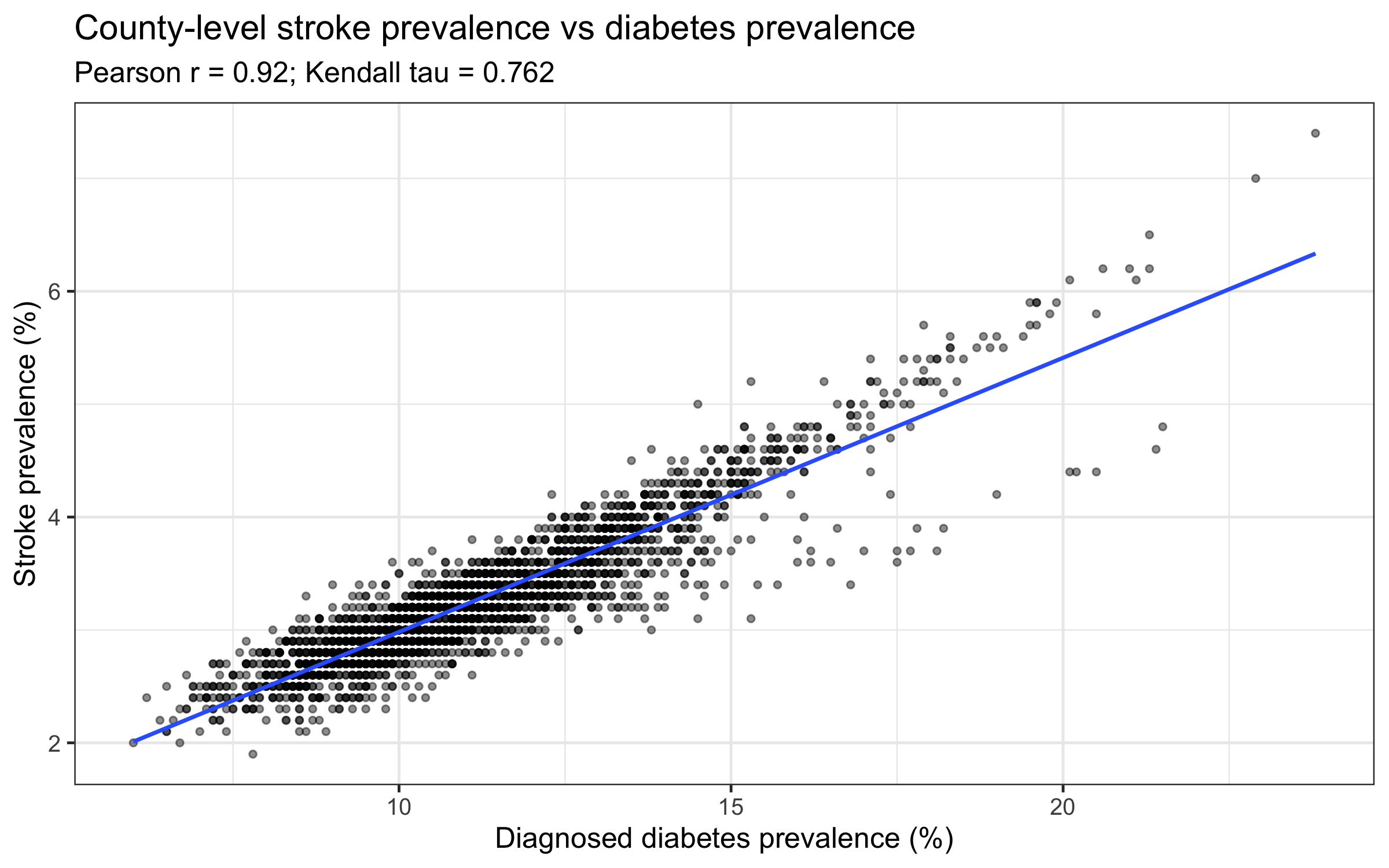}
\caption{Left: empirical distribution of county-level stroke prevalence. The dashed
vertical line marks the threshold $c=4.1\%$. Right: county stroke prevalence versus
diagnosed diabetes prevalence. The strong positive association supports the use of diabetes
prevalence as the auxiliary ranking variable in DUST-MNS.}
\label{fig:analysis_panel}
\end{figure}

\subsection{Theoretical efficiency at the empirical parameter values}\label{ssec:theory_app}

The theoretical efficiency expressions from Sections~\ref{sec:theory}--\ref{sec:effcomp}
can be evaluated directly at the empirical parameter values. In particular, the empirical
estimate $\hat\theta=0.0866$ lies well below the critical thresholds reported in
Table~\ref{tab:thetastar} for the set sizes considered, so the theory predicts a clear
efficiency advantage for DUST-MNS over DUST-SRS in this application. The same empirical
inputs also imply a substantial gain of DUST-SRS over ordinary SRS, reflecting the strong
spatial dependence in county-level stroke prevalence.

Figure~\ref{fig:re_panel} summarizes the main theoretical patterns. The left panel shows
the working-model relative efficiency of imperfect DUST-MNS relative to DUST-SRS as a
function of Kendall's $\tau$ for several set sizes. At the empirical value
$\tau=0.7621$, the efficiency is already close to the perfect-ranking benchmark, so the
loss due to imperfect ranking is modest. The right panel shows
$\mathrm{RE}(\hthet_{\mathrm{DUST\text{-}MNS}},\tthet_{\mathrm{DUST\text{-}SRS}})$ as a function of
set size $k$ at the empirical exceedance rate. Over the range considered, the theoretical
advantage of DUST-MNS remains substantial, although in practice this must be balanced
against the increased ranking difficulty associated with larger $k$.

\begin{figure}[htbp]
\centering
\safeincludegraphics[width=0.48\textwidth]{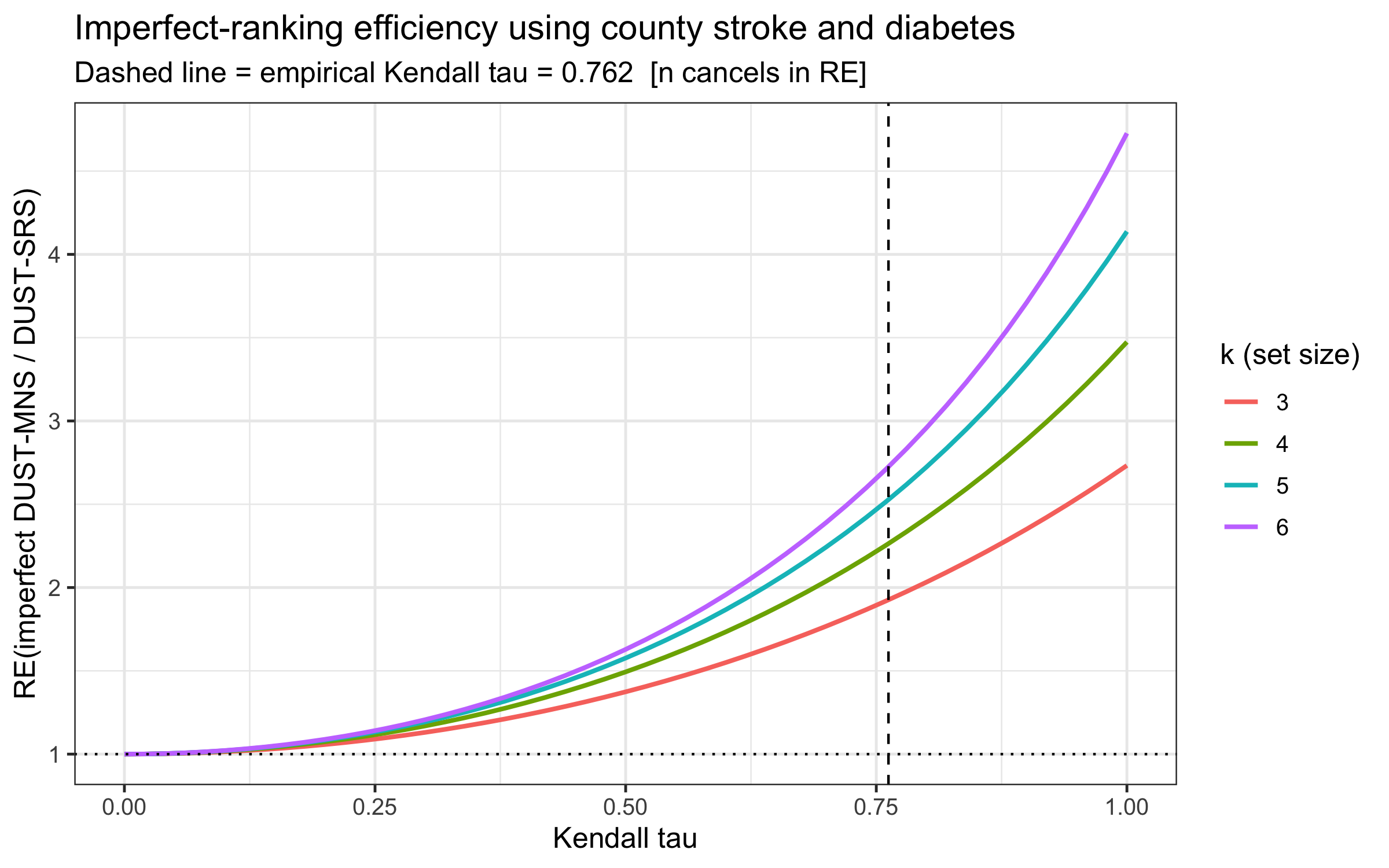}\hfill
\safeincludegraphics[width=0.48\textwidth, height=2.02in]{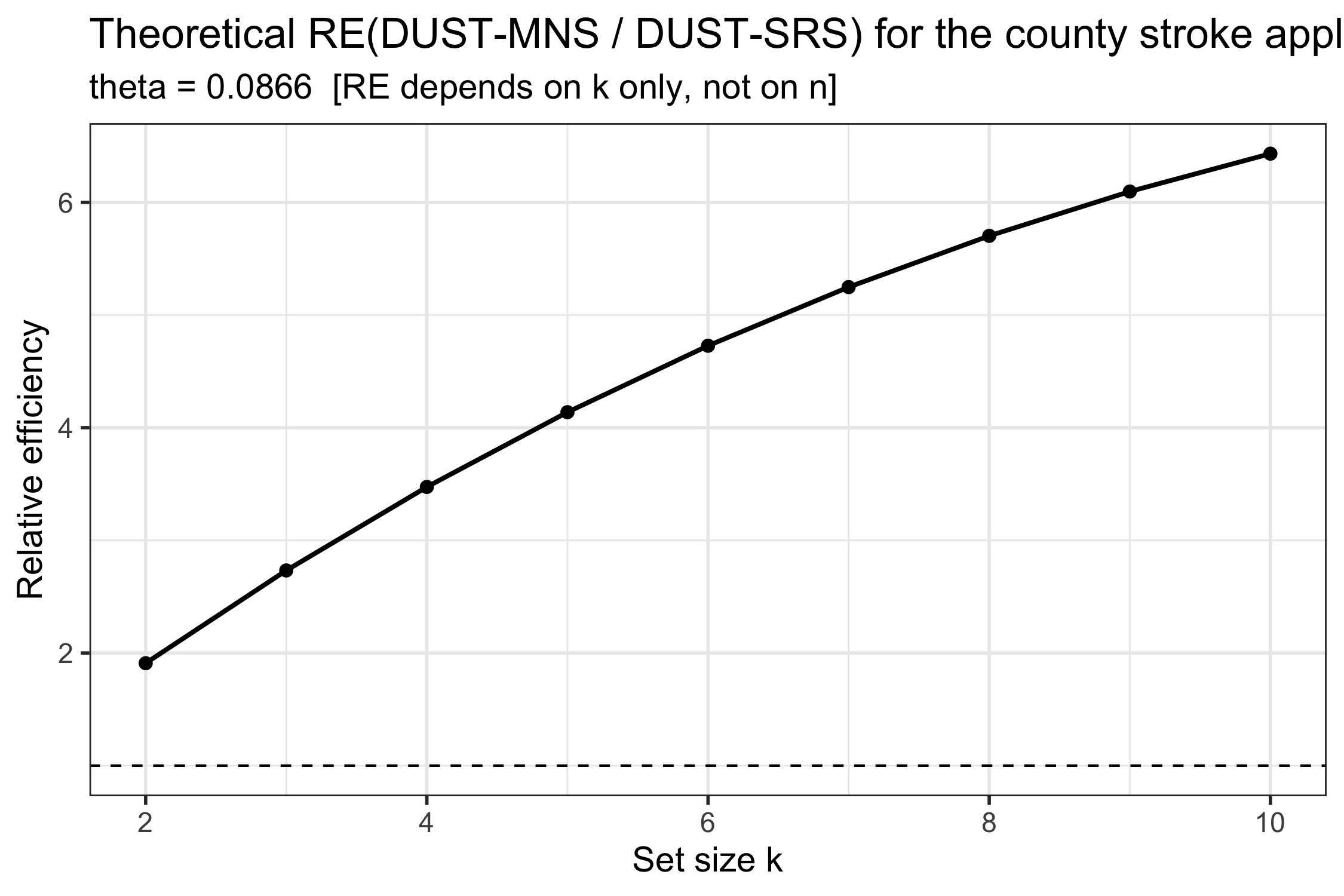}
\caption{Theoretical efficiency at the empirical parameter values. Left: working-model
relative efficiency of imperfect DUST-MNS relative to DUST-SRS as a function of Kendall's
$\tau$ for selected set sizes; the dashed vertical line marks the empirical
$\tau=0.7621$. Right: $\mathrm{RE}(\hthet_{\mathrm{DUST\text{-}MNS}},\tthet_{\mathrm{DUST\text{-}SRS}})$
as a function of set size $k$ at the empirical exceedance rate $\hat\theta=0.0866$.}
\label{fig:re_panel}
\end{figure}

\subsection{Monte Carlo performance evaluation}\label{ssec:mc}

To assess finite-sample performance under the actual county data structure, we conducted a
Monte Carlo study using the observed county-level stroke and diabetes prevalences as a
fixed finite population. In other words, the county frame is held fixed throughout the
experiment and repeated samples are drawn from that empirical population; the Monte Carlo
variation therefore reflects repeated implementation of the competing sampling designs, not
repeated regeneration of a new superpopulation dataset. This complements the model-based
theory by showing how the procedures perform on the observed national county frame.

The simulation grid varies four design inputs: the number of sampled and studied counties
$n\in\{10,20\}$, the set size $k\in\{2,3,4,5\}$, the within-county measurement fraction
$m\in\{0.25\%,0.50\%\}$, and the DUST autocorrelation tuning parameter
$\eta_0\in\{0.15,0.30\}$. The values $n=10$ and $n=20$ represent modest surveillance
budgets; the set sizes are chosen to span the practically relevant range in which ranking
is still credible; and the two $\eta_0$ values are used as design-tuning scenarios rather
than as literal replacements for Moran's $I$. The observed Moran's $I=0.6141$ confirms
substantial positive spatial dependence in the county frame, while $\eta_0=0.15$ and
$0.30$ provide moderate and stronger repulsion settings for the DUST algorithm. For each
configuration we generated $B=5000$ independent Monte Carlo replicates. Within a replicate, the procedures were implemented as follows.
\begin{enumerate}
    \item For SRS, we selected $n$ counties uniformly without replacement from the $N=2921$
    counties.
    \item For DUST-SRS, we selected $n$ counties without replacement using pps-DUST with PPS
    size measure $M_i$ and tuning parameter $\eta_0$.
    \item For DUST-MNS, we first selected $nk$ candidate counties without replacement using
    pps-DUST, then randomly repartitioned those $nk$ counties into $n$ sets of size $k$ in
    every replicate.
    \item Under perfect ranking, the counties within each set were ranked by the true county
    stroke prevalence $p_i$; under imperfect ranking, they were ranked by county diabetes
    prevalence $a_i$.
    \item From each set, only the maximum-ranked county was retained for full measurement,
    yielding $n$ measured counties for each DUST-MNS replicate.
    \item In every sampled county $i$, individual binary outcomes were generated as
    $X_i\sim\mathrm{Binomial}(m_i^{\ast},p_i)$, where
    $m_i^{\ast}=\max\{1,\lfloor f_m N_i\rfloor\}$ and
    $f_m\in\{0.0025,0.0050\}$ is the within-county measurement fraction. The empirical
    exceedance indicator was then formed from the sampled county prevalence
    $X_i/m_i^{\ast}$ relative to the fixed threshold $c=4.1\%$.
    \item For each design, the resulting estimator of the county exceedance proportion was
    compared with the finite-population target
    $\theta_N=N^{-1}\sum_{i=1}^N I(p_i>c)=0.0866$.
\end{enumerate}

Thus, within-county sampling error and between-county design randomness are both present in
the reported Monte Carlo MSE values. Repartitioning the DUST-selected candidate pool in
each replicate prevents the results from being driven by one fixed set construction, while
sampling Bernoulli outcomes within counties reflects the practical fact that county-level
prevalence would typically be estimated from a finite number of measured individuals rather
than observed without error.

Table~\ref{tab:county_mc_eta_combined} reports Monte Carlo mean squared error values
($\times 10^4$) for the county stroke application under the two nonzero DUST tuning
settings, $\eta_0=0.15$ and $\eta_0=0.30$. Results are shown for sample sizes
$n\in\{10,20\}$, set sizes $k\in\{2,3,4,5\}$, within-county measurement fractions
$f_m\in\{0.25\%,0.50\%\}$, and the four competing designs: SRS, DUST-SRS, perfect
DUST-MNS, and imperfect DUST-MNS.

\begin{table}[htbp]
\centering
\footnotesize
\setlength{\tabcolsep}{3.6pt}
\caption{Monte Carlo MSE values ($\times 10^4$) for the county stroke application under the two nonzero DUST autocorrelation settings, $\eta_0=0.15$ and $\eta_0=0.30$. For the perfect and imperfect DUST-MNS estimators, entries are reported as \mbox{MSE (RE)}, where $\mathrm{RE}=\mathrm{MSE}(\text{DUST-SRS})/\mathrm{MSE}(\text{DUST-MNS})$. Thus, values greater than 1 indicate improvement relative to DUST-SRS. Results are shown separately for $f_m=0.25\%$ and $f_m=0.50\%$.}
\label{tab:county_mc_eta_combined}
\begin{tabular}{cc|cccc|cccc}
\toprule
& &
\multicolumn{4}{c|}{$\eta_0=0.15$} &
\multicolumn{4}{c}{$\eta_0=0.30$} \\
\cmidrule(lr){3-6} \cmidrule(lr){7-10}
$n$ & $k$
& SRS
& \makecell{DUST-\\SRS}
& \makecell{Perfect\\DUST-MNS}
& \makecell{Imperfect\\DUST-MNS}
& SRS
& \makecell{DUST-\\SRS}
& \makecell{Perfect\\DUST-MNS}
& \makecell{Imperfect\\DUST-MNS} \\
\midrule
\multicolumn{10}{c}{$f_m = 0.25\%$} \\
\midrule
10 & 2 & 468.55 & 78.63 & 37.40 (2.10) & 36.61 (2.15) & 498.26 & 82.29 & 39.39 (2.09) & 37.07 (2.22) \\
10 & 3 & 488.30 & 82.83 & 27.30 (3.03) & 28.64 (2.89) & 473.95 & 86.69 & 27.28 (3.18) & 27.36 (3.17) \\
10 & 4 & 481.69 & 79.12 & 22.75 (3.48) & 26.37 (3.00) & 481.20 & 89.10 & 23.75 (3.75) & 26.33 (3.38) \\
10 & 5 & 491.14 & 83.35 & 21.86 (3.81) & 26.86 (3.10) & 485.04 & 82.74 & 21.94 (3.77) & 25.54 (3.24) \\
20 & 2 & 379.10 & 40.36 & 19.51 (2.07) & 19.28 (2.09) & 393.47 & 43.82 & 19.74 (2.22) & 18.92 (2.32) \\
20 & 3 & 388.14 & 42.08 & 15.24 (2.76) & 16.68 (2.52) & 402.87 & 41.55 & 13.99 (2.97) & 16.34 (2.54) \\
20 & 4 & 394.88 & 41.37 & 14.09 (2.94) & 17.89 (2.31) & 394.04 & 44.06 & 12.85 (3.43) & 17.54 (2.51) \\
20 & 5 & 375.85 & 43.10 & 14.03 (3.07) & 20.09 (2.15) & 403.01 & 44.91 & 13.49 (3.33) & 18.67 (2.41) \\
\midrule
\multicolumn{10}{c}{$f_m = 0.50\%$} \\
\midrule
10 & 2 & 430.87 & 67.23 & 37.01 (1.82) & 37.04 (1.82) & 433.49 & 68.09 & 34.53 (1.97) & 37.14 (1.83) \\
10 & 3 & 432.58 & 69.47 & 28.92 (2.40) & 30.51 (2.28) & 432.06 & 72.62 & 29.09 (2.50) & 31.09 (2.34) \\
10 & 4 & 437.20 & 71.67 & 25.83 (2.77) & 29.72 (2.41) & 441.95 & 68.28 & 25.41 (2.69) & 30.01 (2.28) \\
10 & 5 & 426.89 & 67.48 & 24.40 (2.77) & 29.33 (2.30) & 428.37 & 70.40 & 23.37 (3.01) & 29.76 (2.37) \\
20 & 2 & 349.44 & 32.48 & 20.58 (1.58) & 22.25 (1.46) & 335.26 & 34.98 & 19.90 (1.76) & 22.31 (1.57) \\
20 & 3 & 342.61 & 36.74 & 17.53 (2.10) & 21.72 (1.69) & 346.34 & 33.90 & 17.47 (1.94) & 20.71 (1.64) \\
20 & 4 & 343.49 & 35.13 & 17.79 (1.97) & 22.69 (1.55) & 336.05 & 34.51 & 16.88 (2.04) & 21.50 (1.61) \\
20 & 5 & 327.56 & 34.97 & 18.06 (1.94) & 24.16 (1.45) & 343.73 & 36.74 & 16.52 (2.22) & 22.95 (1.60) \\
\bottomrule
\end{tabular}
\end{table}

Figure~\ref{fig:mc_mse_bars} summarizes the same Monte Carlo results graphically. Across the simulation grid, the qualitative pattern is stable: SRS has the largest MSE, DUST-SRS improves substantially on SRS, and both versions of DUST-MNS improve further on DUST-SRS. The gap between imperfect and perfect DUST-MNS remains small, reflecting the strong empirical association between stroke prevalence and the auxiliary ranking variable. The strongest gains occur for moderate set sizes $k=3,4,5$, where tail targeting is more aggressive but ranking by diabetes prevalence still remains reliable in the county data.

It is also worth noting that, for fixed set size $k$, the relative efficiency of DUST-MNS with respect to DUST-SRS need not increase with  $n$. At the leading-order variance level, both $\tthet_{\mathrm{DUST\mbox{-}SRS}}$ and $\hthet_{\mathrm{DUST\mbox{-}MNS}}$ decrease at rate $1/n$, so their asymptotic efficiency ratio does not depend on $n$. Indeed,  from \eqref{eq:Delta}
\[
\mathrm{RE}(\hthet_{\mathrm{DUST\mbox{-}MNS}},\tthet_{\mathrm{DUST\mbox{-}SRS}})
\approx
\frac{k^2\theta(1-\theta)^{k-1}}{1-(1-\theta)^k},
\]
which is free of $n$. Thus, increasing $n$ primarily improves absolute precision rather than the relative advantage of one design over the other. The simulation table, however, reports finite-sample relative efficiencies based on mean squared error rather than asymptotic variance alone. Writing
\[
\mathrm{RE}_n(\hthet_{\mathrm{DUST\mbox{-}MNS}},\tthet_{\mathrm{DUST\mbox{-}SRS}})
=
\frac{\mathrm{MSE}(\tthet_{\mathrm{DUST\mbox{-}SRS}})}{\mathrm{MSE}(\hthet_{\mathrm{DUST\mbox{-}MNS}})},
\]
suppose that, for fixed $k$,
\[
\mathrm{MSE}(\tthet_{\mathrm{DUST\mbox{-}SRS}})
=
\frac{A}{n}+\frac{a}{n^2}+o(n^{-2}),
\qquad
\mathrm{MSE}(\hthet_{\mathrm{DUST\mbox{-}MNS}})
=
\frac{B}{n}+\frac{b}{n^2}+o(n^{-2}),
\]
where $A$ and $B$ are the leading variance constants and the $n^{-2}$ terms collect lower-order contributions, including squared bias. Then
\[
\mathrm{RE}_n(\hthet_{\mathrm{DUST\mbox{-}MNS}},\tthet_{\mathrm{DUST\mbox{-}SRS}})
=
\frac{A/n+a/n^2+o(n^{-2})}{B/n+b/n^2+o(n^{-2})}
=
\frac{A}{B}
\left[
1+\frac{a/A-b/B}{n}+o(n^{-1})
\right].
\]
Therefore, for fixed $k$, the empirical relative efficiency may increase or decrease slightly with $n$ depending on the lower-order terms, even though its leading asymptotic value is constant in $n$.

\begin{figure}[htbp]
\centering
\safeincludegraphics[width=0.99\textwidth]{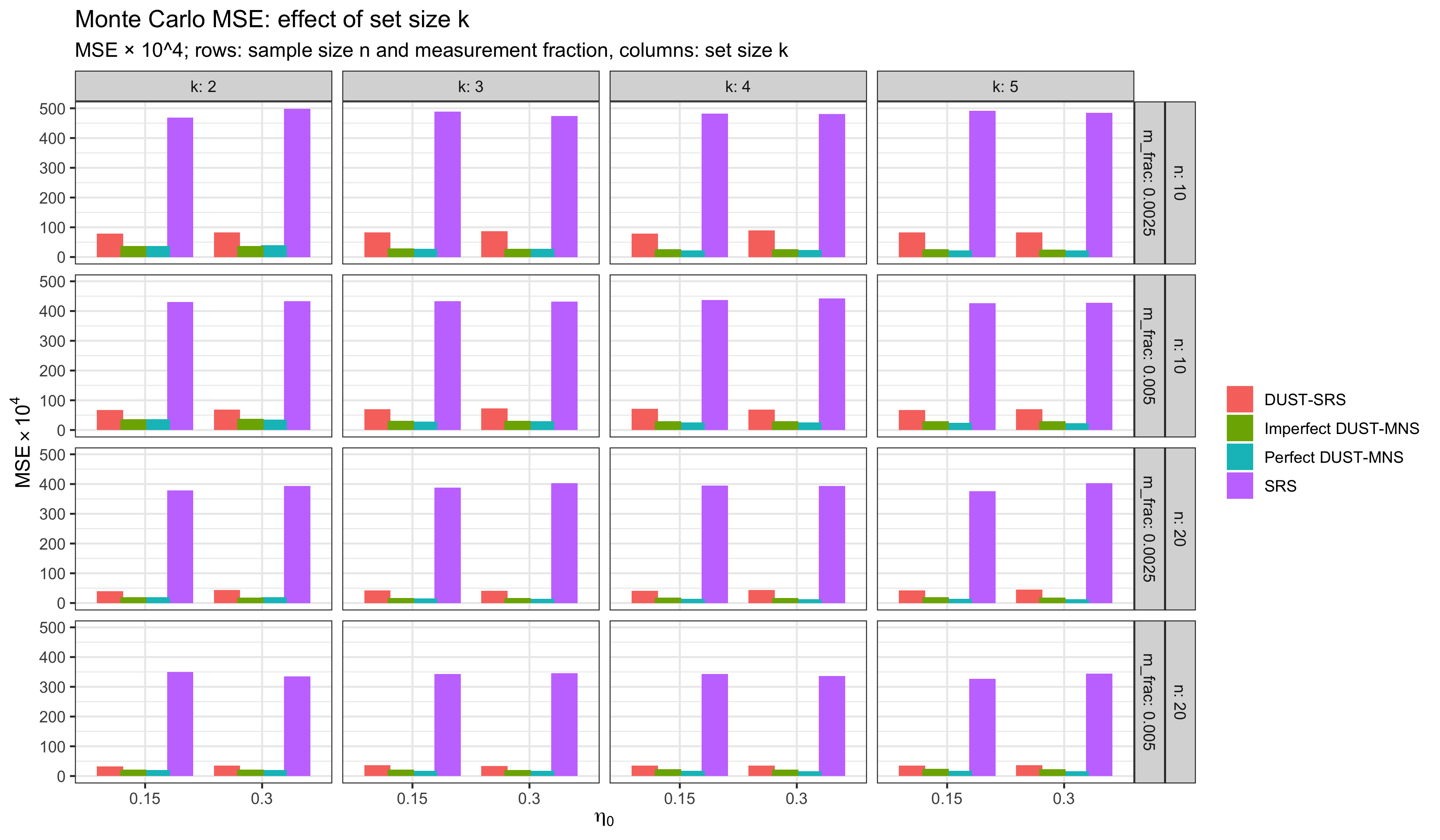}
\caption{Comparison of Monte Carlo MSE values ($\times 10^4$) for the county stroke application under the two nonzero DUST autocorrelation settings, $\eta_0=0.15$ and $\eta_0=0.30$, two sample sizes $n\in \{10, 20\}$ and set sizes $k\in\{2, 3, 4, 5\}$  (for the MNS design)  for the county stroke
application. Within each panel, bars compare DUST-SRS, imperfect DUST-MNS, perfect
DUST-MNS, and SRS.}
\label{fig:mc_mse_bars}
\end{figure}

Overall, the empirical evidence supports the theoretical message of the paper. For
county-level stroke surveillance, DUST-MNS combines the benefit of spatial spreading with
the benefit of upper-tail targeting. The application also shows that, when a strong
auxiliary ranking variable is available, the imperfect-ranking version of the design can
perform nearly as well as the perfect-ranking benchmark.

%======================================================================
\section{Conclusion}\label{sec:conclusion}
%======================================================================

We have introduced a maxima-nominated sampling framework for estimating the threshold
exceedance proportion $\theta=\Pr(p_i>c)$ in spatially correlated areal data. The central
idea is to combine two design principles: spatial spreading, achieved through pps-DUST,
and upper-tail targeting, achieved through maxima nomination within small candidate sets.
Together, these features yield a sampling design tailored to threshold-based inference in
settings where nearby areas are correlated and only a minority of areas exceed the policy
benchmark. The resulting estimator has a simple closed-form.  Under the working model,
DUST-SRS improves on ordinary SRS by reducing the inflation due to spatial dependence,
while DUST-MNS can further improve on DUST-SRS in the moderate-to-rare exceedance regime.
The critical threshold $\theta^\star(k)$ depends on the set size $k$ and determines when
this additional efficiency gain is achieved.

Our asymptotic analysis is carried out in the practically relevant regime in which the set
size $k$ is fixed and the number of sampled sets $n$ increases. In that setting, the
calibrated MNS estimator admits a delta-method variance approximation and a simple
bias-corrected version. The theory also clarifies the tradeoff inherent in the choice of
$k$: increasing $k$ strengthens the tail-targeting effect of maxima nomination, but it
also increases ranking difficulty and finite-sample bias. This supports the use of small
set sizes, such as $k=3,4,$ or $5$, with precision controlled primarily through $n$.

The imperfect-ranking extension shows how the method can be calibrated when the maximum is
identified using an auxiliary variable rather than the unobservable $p_i$. In practice, the
quality of this auxiliary ranking is crucial. The county-level stroke application
illustrates that when the auxiliary variable is strongly associated with the study
variable, the imperfect-ranking version of DUST-MNS can remain close to the ideal
benchmark while still being feasible in practice.

Several directions remain open. These include multivariate exceedance problems, adaptive
choice of the threshold $c$, refined models for imperfect ranking, and integration with
spatial regression structures for $\theta$. More broadly, the framework developed here
suggests that rank-based designs can be adapted in a principled way to modern
threshold-focused spatial surveillance problems, where the main inferential target lies in
the upper tail rather than at the mean.

\section*{Acknowledgment}  I would like to  gratefully acknowledge the  partial support from  NSERC.

%%%%%%%%%%%%%%%%%%%%%%%%%%%%%%%%%%%%%%%%%%%%%%%%%%%
	%
	%  Section
	% 
	%%%%%%%%%%%%%%%%%%%%%%%%%%%%%%%%%%%%%%%%%%%%%%%%%
	
	\section*{Data Availability statement} 
The data underlying this article are publicly available from sources cited in Section~\ref{sec:application}. County-level stroke prevalence estimates were obtained from CDC PLACES, and county geographic boundary files were obtained from the U.S. Census Bureau cartographic boundary shapefiles. The analytic data set and code used to reproduce the results, tables, and figures are available from the corresponding author upon request.

\end{document}